\documentclass{article}

\PassOptionsToPackage{numbers, compress}{natbib}

\usepackage[preprint]{neurips_2026}

\usepackage{amsthm}
\usepackage{amsmath}
\usepackage{amssymb}
\usepackage[utf8]{inputenc} 
\usepackage[T1]{fontenc}    
\usepackage{hyperref}       
\usepackage{url}            
\usepackage{booktabs}       
\usepackage{amsfonts}       
\usepackage{nicefrac}       
\usepackage{microtype}      
\usepackage{xcolor}         
\usepackage{wrapfig}

\usepackage{tikz}
\usetikzlibrary{matrix,decorations.pathreplacing}
\usetikzlibrary{positioning,fit,arrows.meta,calc,matrix}
\usetikzlibrary{positioning,calc,fit,backgrounds}
\usepackage{pgfplots}
\pgfplotsset{compat=1.18}
\usepackage{booktabs}
\usepackage{wrapfig}

\newtheorem{theorem}{Theorem}[section]
\newtheorem{lemma}[theorem]{Lemma}

\usepackage[ruled,vlined]{algorithm2e}

\usepackage{natbib}
\usepackage{comment}
\usepackage{cleveref}

\usepackage{mathtools}
\usepackage{todonotes}

\theoremstyle{definition}

\theoremstyle{definition}
\newtheorem{definition}{Definition}

\newtheorem{remark}{Remark}
\newtheorem{example}{Example}
\newtheorem*{lemmarestatement}{Lemma}
\newtheorem*{theoremrestatement}{Theorem}

\usepackage{amsmath,amsfonts,bm}

\def\eqref#1{equation~\ref{#1}}

\def\1{\bm{1}}

\DeclareMathAlphabet{\mathsfit}{\encodingdefault}{\sfdefault}{m}{sl}
\SetMathAlphabet{\mathsfit}{bold}{\encodingdefault}{\sfdefault}{bx}{n}

\newcommand{\E}{\mathbb{E}}

\newcommand{\XSpace}{\mathcal{X}}

\newcommand{\dd}{\,\mathrm{d}}
\title{Multi-Marginal Couplings for Metropolis--Hastings}

\author{%
    Buu Phan \\
   University of Toronto \\
   \texttt{truong.phan@mail.utoronto.ca}
   \And
   Gergely Flamich \\
   Imperial College London \\
   \texttt{g.flamich@imperial.ac.uk}
   \AND
   Ashish Khisti \\
   University of Toronto \\
   \texttt{akhisti@ece.utoronto.ca}
   \And
   Shahab Asoodeh \\
   McMaster University \\
   \texttt{asoodeh@mcmaster.ca}
}

\begin{document}

\maketitle

\begin{abstract}
Convergence diagnosis for Markov chain Monte Carlo is a matter of fundamental importance in computational statistics: it determines the resources allocated to a particular sampling problem and influences the practitioner's view of the quality of estimates obtained from a Markov chain. Motivated by this, we contribute to the emerging class of coupling-based convergence diagnostic algorithms.
Concretely, we study coupling multiple Metropolis--Hastings chains using multi-marginal coupling. We introduce a natural objective for this setting and establish lower and upper bounds by drawing connections to list-level distribution coupling and distributed pairwise-matching problems. This analysis ultimately leads to a shared-randomness Poisson Monte Carlo
construction for coupling multiple Markov chains.  In this process, we avoid a key dimension-dependent bottleneck in the runtime complexity of classical Poisson Monte Carlo by developing an adaptive rule for updating the point process, yielding significant gains in high-dimensional settings. Experiments on grand couplings of Markov chains show that our methods improve coalescence rates across dimensions, reducing meeting times by up to \(50\%\) compared with existing baselines.
\end{abstract}
\section{Introduction}
Markov chain Monte Carlo (MCMC) is a central methodology for sampling from complex probability distributions when exact simulation techniques are computationally intractable. It proceeds by constructing a Markov chain whose stationary distribution matches a target distribution of interest, allowing expectations under the target law to be approximated from the chain trajectory. Because of its broad applicability in Bayesian inference, statistical physics, and machine learning, MCMC has become one of the most widely used tools for high-dimensional probabilistic computation.

A fundamental challenge in MCMC is convergence diagnosis: determining whether the chain has run long enough for its distribution to be sufficiently close to stationarity. Classical approaches often rely on running multiple independent chains with different initializations and comparing their behavior, as in the Gelman--Rubin diagnostic and its extensions \cite{brooks1998general,gelman1992inference}. Coupling provides another principled alternative by jointly evolving chains from different initializations and using their meeting behavior as a diagnostic signal, where a meeting occurs when two or more chains coalesce in the same state \cite{biswas2019estimating,corenflos2025coupling}. As the effectiveness of such diagnostics depends on how quickly coalescence occurs, several studies have examined the design of efficient couplings, with a primary focus on two chains \cite{wang2021maximal}. By comparison, the multi-chain coupling setting has received relatively limited attention, despite the existence of convergence diagnostics based on grand couplings, which jointly evolve chains from multiple initial states until they coalesce \cite{johnson1996studying}. 
\par
Coupling multiple chains is not as simple as pairwise coupling. 
A natural first strategy is to construct a multi-chain coupling using pairwise coupling as a building block.
For example, we could have a ``server-client'' topology, in which we couple each chain to a distinguished reference chain;
alternatively, we could randomly choose coupling partners at each transition \cite{wang2021maximal,biswas2019estimating}.
Such approaches, however, can be highly suboptimal: they may spend effort attempting to couple chains whose states are far apart, and thus whose coalescence is unlikely. 
A more effective strategy should preferentially match chains that are close under the transition kernel, but identifying such pairs is challenging in high-dimensional settings. 
In this work, we develop adaptive coupling strategies in which the joint construction itself determines which chains should meet by connecting the multi-marginal setting to list-level coupling and distributed matching. 
This strategy avoids committing to fixed partners, promotes coalescence among compatible chains, and leads to a kernel-coupling construction with effective meeting behavior and favorable runtime scaling in high-dimensional multi-chain settings. In particular, our contributions are:
\begin{enumerate}
\item In \Cref{sec:multi-marginal}, we study multi-chain coupling from a multi-marginal perspective and analyze theoretical properties of the resulting meeting behavior through list-level coupling (membership inclusion), and distributed matching formulations.
\item In \Cref{sec:mh_pml}, we introduce an efficient Poisson-matching scheme for coupling multiple Metropolis--Hastings kernels. The scheme uses a flexible adaptive proposal update that yields runtime scaling primarily with the number of chains rather than the dimension.
\item In \Cref{sec:experiments}, we show experimentally that our method outperforms natural grand-coupling baselines, substantially reducing coalescence times in high-dimensional settings when using many chains.
\end{enumerate}

\section{Algorithms for Multi-Marginals Coupling }\label{sec:multi-marginal}

Before discussing the multi-marginal formulation, we recall the pairwise case. 
For probability measures \(P\) and \(Q\), the largest possible agreement probability over all couplings,
\(
\Pr(X^P=X^Q),
\)
is \(1-d_{\mathrm{TV}}(P,Q)\), where $d_{\mathrm{TV}}(P,Q)$ is the total-variation distance. A coupling attaining this value is called a maximal coupling \citep{levin2017markov}; we review its achievability in Appendix~\ref{appendix:max_coupling}.

When there are more than two marginals, there is no canonical notion of agreement. Consider \(C\) random variables \(X^1,\dots,X^C\) with marginals \(X^i\sim P^i\). Natural objectives could be to maximize the probability of full coalescence or the probability of at least one pairwise match, as studied in the multi-marginal coupling setting of \citet{angel2019pairwise}:
\[
\Pr(X^1=X^2=\cdots=X^C)
\qquad \text{or} \qquad
\Pr(\exists\, i\neq j : X^i=X^j).
\]
For coupling multiple Markov chains, coalescence is expected to occur progressively over iterations. Therefore, full coalescence is often too stringent, especially as \(C\) grows. In contrast, the existence of a single pairwise match does not distinguish between a coupling that merges two variables and one that merges many. These observations motivate the following measure of partial coalescence.

\paragraph{Problem setup.}
To quantify partial coalescence, we consider the optimization problem
\begin{equation}
\mathbf G^*
=
\inf_{\gamma\in\Gamma(P^1,\dots,P^C)}
\mathbb E_{\gamma}[G],
\qquad
G
\triangleq
\bigl|\{\text{unique values among }X^1,\dots,X^C\}\bigr|,
\label{opt_obj}
\end{equation}
where \(\Gamma(P^1,\dots,P^C)\) is the set of all couplings of \(P^1,\dots,P^C\) and \(G\) is the number of distinct realized values among the \(C\) variables. It directly measures progressive coalescence: \(G{=}C\) corresponds to complete separation, \(G{=}1\) to full coalescence, and intermediate values to partial merges.

The structure of an optimal coupling for \eqref{opt_obj}, and how to construct one, are not immediately clear. A natural first attempt is to extend maximal coupling by selecting one variable as an anchor and coupling all other variables to it. However, this construction is asymmetric and designed to maximize agreement with the chosen anchor, rather than directly minimize the number of distinct realized values. This section will show that this objective admits the following informative lower and upper bounds, shown Theorem \ref{theorem:general_bounds}.\\

\begin{theorem}\label{theorem:general_bounds}
Given probability measures \(P^1,\dots,P^C\), and a permutation $\sigma$ over $C$ elements, define
\[
\bar{\nu}^{\sigma (k)}(\mathrm dx)
\coloneqq
\frac{1}{C-k}\sum_{j=k+1}^{C} P^{\sigma(j)}(\mathrm dx),
\qquad k=1,\dots,C-1 .
\]
Moreover, let \(E_m(P\|Q)=\int \max\{P(x)-m Q(x),0\}\,\mathrm dx\) denote the Hockey-Stick divergence \citep{sason2016f} with $m \geq 1$.
Then, the optimal value \(\mathbf G^*\) in \eqref{opt_obj} satisfies
\begin{equation}
1+ \sup_{\sigma }\sum_{k=1}^{C-1}
E_{C-k}\!\left(P^{\sigma(k)}\middle\|\bar{\nu}^{\sigma (k)}\right)
\leq
\mathbf G^*
\leq
\frac{
1+\sqrt{
1+8\sum_{1\le i<j\le C}
\frac{2\,d_{\mathrm{TV}}(P^i,P^j)}
{1+d_{\mathrm{TV}}(P^i,P^j)}
}
}{2}.
\label{prop:bnd_general}
\end{equation}
\end{theorem}
Both the lower and upper bounds are tight at the two extreme configurations. If \(X^1=\cdots=X^C\) almost surely, then \(d_{\mathrm{TV}}(P^i,P^j)=0\) for all \(i\neq j\) and \(G=1\), so both bounds hold with equality. Conversely, if \(X^i\neq X^j\) almost surely for all \(i\neq j\), then \(G=C\), and both bounds are again tight. Beyond these extremes, the lower bound is tight when \(C=2\), in which case the problem reduces to the maximal coupling of two variables, and in certain other special cases; see Example \ref{example:optimal_list} and Appendix \ref{example:pml_shifted_exp}. We now establish  Theorem~\ref{theorem:general_bounds} and develop coupling constructions motivated by the structure of these bounds. 

\subsection{Lower Bound with List-Level Coupling}\label{sec:list_coupling_method}
To establish the lower bound in Theorem~\ref{theorem:general_bounds}, we first rewrite the expected cluster count in terms of list-level inclusion
probabilities. 
Fixing any ordering
\(\sigma(1),\dots,\sigma(C)\) of the indices, we show in
Appendix~\ref{appendix:list} that, for any coupling of \(X^1,\dots,X^C\),
\begin{equation}
\begin{aligned}
G
&=
C-\sum_{k=1}^{C-1}
\mathbf{1}\!\left(
X^{\sigma(k)} \in
\{X^{\sigma(k+1)},\dots,X^{\sigma(C)}\}
\right)
\end{aligned}
\end{equation}
Intuitively, scanning the ordering backward, the inclusion event means the current value has
already appeared in the suffix, so it does not create a new cluster. Furthermore, this idea yields a natural coupling procedure in the multi-marginal setting, where the matching partner need not be fixed in advance. In particular, under a given ordering, each variable may coalesce with any later variable that realizes the same value. Thus, we get the following lower bound on the optimal objective \(\mathbf G^\star\):
\begin{align}
   \mathbf{G}^* =  \inf_{\gamma \in \Gamma} \mathbb E_{\gamma}[G] 
&\geq C - \sum_{k=1}^{C-1}\sup_{\gamma \in \Gamma } 
\Pr\!\left(
X^{\sigma(k)} \in
\{X^{\sigma(k+1)},\dots,X^{\sigma(C)}\}
\right), \label{relax_ineq}
\end{align}
where $\Gamma \triangleq \Gamma (P^1,\cdots,P^C)$ and the inequality in (\ref{relax_ineq}) relaxes the problem by optimizing each list-level inclusion (or membership) probability separately; note that a different coupling may attain the optimum for each term. The following maximal list-level coupling result upper bounds each separated
inclusion probability. Together with \eqref{relax_ineq}, it directly yields the lower bound in Theorem~\ref{theorem:general_bounds}.

\begin{theorem}[List-Level Coupling]\label{theorem:list_coupling}
Let \(X\sim\mu\) and \(Y^1,\dots,Y^m\) have marginals \(\nu^1,\dots,\nu^m\), all absolutely continuous w.r.t. a common dominating measure $\lambda$. Then we have:
\[
\sup_{\gamma\in\Gamma(\mu,\nu^1,\dots,\nu^m)}
\Pr_{\gamma}\!\left(X\in\{Y^1,\dots,Y^m\}\right)
=
1-E_m(\mu\|\bar\nu),
\]
where \(\bar\nu \coloneqq m^{-1}\sum_{j=1}^m \nu^j\) is the mixture barycenter of \(\nu^1,\dots,\nu^m\).
\end{theorem}
Further details and proof of this theorem are given in Appendix~\ref{proof_list}. When \(m=1\), the result
reduces to the classical maximal coupling identity. To the best of our
knowledge, this provides the first list-level coupling analog of maximal
coupling. The supremum is attained by Algorithm~\ref{alg:list_coupling_barycenter},
together with the residual-sampling subroutine in
Algorithm~\ref{alg:residual_sampler}. The construction samples \(X\sim\mu\),
then on the coupling event randomly chooses which \(Y^j\) coalesces with \(X\),
and samples the remaining \(Y\)'s from their residual distributions. Specifically, the residual density of
\(R^j\) is
\begin{equation}
  r^j(x)
=
\frac{
\nu^j(x)
}{
1-\alpha^j
}\left(1-\min\left\{1,\frac{\mu(x)}{m\bar\nu(x)}\right\}\right)
,
\qquad
\alpha^j
=
\int
\nu^j(x)\min\left\{1,\frac{\mu(x)}{m\bar\nu(x)}\right\}
\,\lambda(\mathrm dx),  \label{residual_distribution}
\end{equation}
with the convention that the integrand is \(0\) whenever \(\bar\nu(x){=}0\).%
\footnote{With a slight abuse of notation, we denote their densities with the same letter when no confusion arises.}
Note that the sampling steps highlighted in \textcolor{blue}{blue} in Algorithm~\ref{alg:list_coupling_barycenter} need not be independent. 
Hence, when direct residual sampling is available, such as when the distributions involved are discrete, there is a natural greedy strategy: after each list-level match is formed, remove the matched variables and recursively apply Algorithm~\ref{alg:list_coupling_barycenter} to the remaining unmatched list. For general cases, see Remark \ref{list_technical_challenge} below Algorithm \ref{alg:residual_sampler}.  Although the greedy strategy does not attain the lower bound in general, it can be optimal in certain settings; see Example~\ref{example:optimal_list}. Empirically, we find that it outperforms fixed-reference baselines, see Section \ref{sec:experiments}.

\begin{example}\label{example:optimal_list}
Suppose $P^1,P^2,P^3,P^4$ decompose into two separated overlap pairs, $P^1$ with $P^2$
and $P^3$ with $P^4$, with no cross-overlap between the two pairs, then the
recursive list-level coupling attains maximal pairwise matching within both
pairs, thereby achieving optimal performance. The lower bound in Theorem \ref{theorem:general_bounds} is also achieved for a similar reason.
\end{example}
\begin{figure}[t]
\centering

\begin{minipage}[t]{0.48\linewidth}
\begin{algorithm}[H]
\SetAlgoLined
\DontPrintSemicolon
\SetKwFunction{Sample}{Sample}

\textbf{Input:} densities \(\mu,\nu^1,\dots,\nu^m\); residual samplers \(R^1,\dots,R^m\)\;
\textbf{Output:} \((X,Y^1,\dots,Y^m)\)\;

\(\bar\nu(x)\gets m^{-1}\sum_{j=1}^m \nu^j(x)\)\;
\(X,U\gets \Sample(\mu),\Sample(\mathrm{Unif}(0,1))\)\;

\If{\(U\le \min\{1,m\bar\nu(X)/\mu(X)\}\)}{
    Sample \(I\in[m]\) with
    \(\Pr(I=j\mid X)=\nu^j(X)/(m\bar\nu(X))\)\;
    \(Y^I\gets X\)\;
    \textcolor{blue}{\(Y^j\gets \Sample(R^j)\) for  \(j\in[m]\setminus\{I\}\)\;}
}
\Else{
   \textcolor{blue}{\(Y^j\gets \Sample(R^j)\) for all \(j\in[m]\)\;} 
}

\Return{}{\((X,Y^1,\dots,Y^m)\) , matched index \(I\)}\;
\caption{Maximal list coupling}
\label{alg:list_coupling_barycenter}
\end{algorithm}
\end{minipage}
\hfill
\begin{minipage}[t]{0.48\linewidth}
\begin{algorithm}[H]
\SetAlgoLined
\DontPrintSemicolon
\SetKwFunction{Sample}{Sample}

\textbf{Input:} densities \(\nu^j,\mu,\bar\nu\)\;
\textbf{Output:} \(Y^j\sim R^j\)\;

\Repeat{\(U\le a(Z)\)}{
    \(Z,U{\gets} \Sample(\nu^j),\Sample(\mathrm{Unif}(0,1))\)
    \(a(Z)\gets \left(1-\frac{\mu(Z)}{m\bar\nu(Z)}\right)_+\)\;
}

\Return{\(Z\)}\;
\caption{Residual sampler for \(R^j\)}
\label{alg:residual_sampler}
\end{algorithm}
\begin{remark}\label{list_technical_challenge}
   In general, it is nontrivial to recursively apply the \textcolor{blue}{blue}
step together with maximal list coupling when residual sampling is required, e.g. when all measures are Gaussians. 
Without the normalizing constant \(1-\alpha^j\) of \(r^j\) in
\eqref{residual_distribution}, the barycenter likelihood cannot be computed
explicitly.
\end{remark}
\end{minipage}

\end{figure}
\vspace{-3pt}
\subsection{Upper Bound via Poisson Matching and Fast Coupling Construction}\label{sec:accel_pml}
Although list-level coupling provides useful theoretical guarantees and suggests a natural adaptive strategy, its implementation for general continuous measures remains challenging. We now introduce an alternative construction based on a shared marked Poisson process that preserves the idea of Algorithm~\ref{alg:list_coupling_barycenter}: the matching partners are determined adaptively. Related constructions were also considered by \citet{angel2019pairwise}, but computationally, their use was limited to special discrete settings, since direct implementation becomes prohibitively expensive in high dimensions. Here, we introduce a new technique that turns this Poisson-process construction into a practical coupling algorithm for continuous high-dimensional MCMC. The induced pairwise coupling probabilities provide the key control needed to establish the upper bound in Theorem~\ref{theorem:general_bounds}.
\paragraph{Poisson Monte Carlo (PMC) .} We begin by introducing PMC, a sampling method for drawing samples from a target distribution $P$ using a proposal distribution $\mu$.
Assume $P \ll \mu$\footnote{In the context of PMC, this distribution $\mu$ serves as a proposal measure and should not be confused with the reference distribution used in the list-level coupling of Algorithm~\ref{alg:list_coupling_barycenter}.}, and let
\(
\Omega = \{(S_i,X_i)\}_{i=1}^\infty
\)
be a marked Poisson process, where $\{S_i\}_{i=1}^\infty$ are the arrival times of a unit-rate Poisson process on $\mathbb{R}_+$ and the marks $X_i$ are i.i.d. samples from $\mu$. Define the random index 
\begin{align}
J
=
\arg\min_i \; S_i
\left( \frac{\dd P}{\dd \mu}(X_i) \right)^{-1},
\qquad
\text{where } \frac{\dd P}{\dd\mu} \text{ is the Radon--Nikodym derivative.}
\label{PP_opt}
\end{align}
Then, setting $X^P = X_{J}$, we have $X^P \sim P$ (marginalize over $\Omega$). The procedure to solve the above optimization problem is presented in Appendix \ref{app:solve_PP_opt}, following \cite{maddison2014sampling,theis2022algorithms,flamich2024data}. To couple \(X^1,\dots,X^C\) simultaneously, we let all sampler parties use a common Poisson process \(\Omega\), which induces the pairwise matching guarantees, known as the Poisson Matching Lemma (PML) \cite{li2021unified}.

\begin{lemma}[Poisson Matching Lemma \cite{li2021unified}] \label{lemma_pml}
Let \(\Omega = \{(S_i,X_i)\}_{i=1}^\infty\) be a marked Poisson process on \(\mathbb{R}_+\times\XSpace\) with intesity measure \(\mathrm ds\otimes \mu\),  and \(P\) and \(Q\) be two probability measures such that \(P,Q \ll \mu\). Let \(P\) and \(Q\) be coupled via the Poisson Monte Carlo procedure described above, producing samples \(X^P\sim P\) and \(X^Q\sim Q\). 
Then, we have:
\begin{align}
    \Pr(X^P=X^Q|X^P=x) \geq \left(1 + \frac{\dd P}{\dd Q}(x)\right)^{-1} , 
    \qquad 
    \Pr(X^P = X^Q) \geq \frac{1 - d_{\mathrm{TV}}{ (P,Q)}}{1 + d_{\mathrm{TV}} {(P,Q)}}, \label{pml_bound_tv}
\end{align}
where the gap in matching probability compared with maximal coupling is small in practice~\cite{daliri2025coupling,li2021unified}. 
\end{lemma}
A detailed proof showing how the PML construction induces the upper bound in Theorem~\ref{theorem:general_bounds} is provided in Appendix~\ref{appendix:proof_bnd_G}. Here \(\mathrm dP/\mathrm dQ\) is interpreted as the extended ratio
\((\mathrm dP/\mathrm d\mu)/(\mathrm dQ/\mathrm d\mu)\), with value \(+\infty\)
when the denominator vanishes but the numerator does not.  We now discuss the runtime complexity of PML and present a new technique to accelerate the coupling process. 
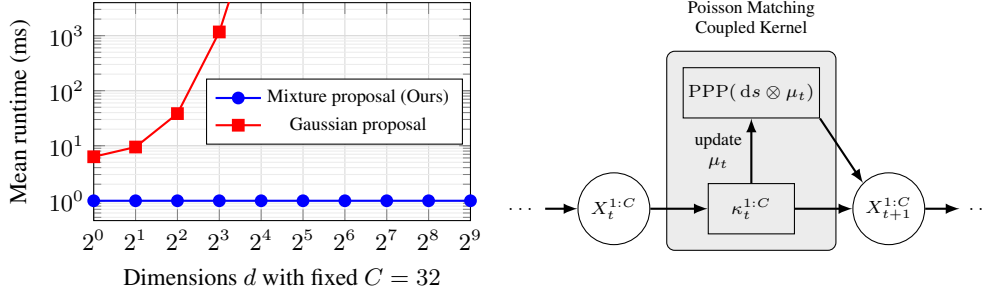
\begin{figure}[t]
\centering
\hspace{-30pt}
\begin{minipage}[c]{0.47\textwidth}
\centering
\begin{tikzpicture}
\begin{axis}[
    width=\textwidth,
    height=0.68\textwidth,
    xlabel={Dimensions $d$ with fixed $C=32$},
    ylabel={Mean runtime (ms)},
    xmode=log,
    ymode=log,
    log basis x=2,
    xmin=1, xmax=512,
    ymax=10^3.6,
    restrict y to domain*=0:10^3.1,
    unbounded coords=discard,
    grid=both,
    major grid style={gray!30},
    minor grid style={gray!15},
    xtick={1,2,4,8,16,32,64,128,256,512},
    xticklabels={$2^0$,$2^1$,$2^2$,$2^3$,$2^4$,$2^5$,$2^6$,$2^7$,$2^8$,$2^9$},
    tick label style={font=\footnotesize},
    label style={font=\footnotesize},
    legend style={
        at={(0.98,0.5)},
        anchor=east,
        font=\scriptsize
    },
    mark options={solid},
]
\addplot[
    thick,
    color=blue,
    mark=*,
    mark options={fill=blue},
] coordinates {
    (1,   0.617)
    (2,   0.605)
    (4,   0.640)
    (8,   0.661)
    (16,  0.616)
    (32,  0.626)
    (64,  0.601)
    (128, 0.639)
    (256, 0.627)
    (512, 0.702)
};
\addlegendentry{Mixture proposal (Ours)}

\addplot[
    thick,
    color=red,
    mark=square*,
    mark options={fill=red},
] coordinates {
    (1, 6.302661172812805)
    (2, 9.476135837612674)
    (4, 38.3550911443308)
    (8, 1167.7786182542332)
    (10, 6739.802)
};
\addlegendentry{Gaussian proposal}
\end{axis}
\end{tikzpicture}
\end{minipage}
\hspace{-3pt}
\begin{minipage}[c]{0.42\textwidth}
\centering
\vspace{-20pt}
\begin{tikzpicture}[
    >=latex,
    every node/.style={font=\scriptsize},
    state/.style={draw, circle, inner sep=1pt, minimum size=1.0cm},
    kernel/.style={draw, rectangle, inner sep=2pt, minimum width=1.2cm, minimum height=0.7cm},
    omega/.style={draw, rectangle, inner sep=2pt, minimum width=1.2cm, minimum height=0.7cm},
    groupbox/.style={draw, rounded corners=3pt, fill=gray!15, inner sep=6pt},
    scale=0.95,
    transform shape
]

\node (ldotsL) at (0,0) {$\dots$};
\node[state, right=0.45cm of ldotsL] (Xt) {$X_t^{1:C}$};
\node[kernel, right=0.8cm of Xt] (Kt) {$\kappa_t^{1:C}$};
\node[state, right=0.8cm of Kt] (Xtp) {$X_{t+1}^{1:C}$};
\node (ldotsR) at ($(Xtp.east)+(0.8cm,0)$) {$\dots$};

\node[omega, above=0.9cm of Kt] (Omega) {$ \operatorname{PPP}(\dd s \otimes \mu_t)
$};

\begin{pgfonlayer}{background}
\node[groupbox, fit=(Kt)(Omega)] (PMbox) {};
\end{pgfonlayer}

\node[font=\scriptsize, align=center, above=2pt of PMbox.north]
{Poisson Matching\\Coupled Kernel};

\draw[->, thick] (ldotsL.east) -- (Xt.west);
\draw[->, thick] (Xt.east) -- (Kt.west);
\draw[->, thick] (Kt.east) -- (Xtp.west);
\draw[->, thick] (Xtp.east) -- (ldotsR.west);

\draw[->, thick] (Kt.north) -- node[midway, left, align=center] {update\\$\mu_t$} (Omega.south);
\draw[->, thick] (Omega.south east) -- (Xtp.north west);

\end{tikzpicture}
\end{minipage}

\caption{Poisson point process (PPP) coupling. (Left) Runtime comparison across dimensions between our constructed proposal and a naive strategy. (Right) Schematic illustration of how our method applies to multi-chain Markov coupling through adaptive updates of the shared PPP.}
\vspace{-10pt}
\label{fig:runtime_and_schematic}
\end{figure}

\paragraph{Runtime Complexity Problem.}One caveat of the PMC approach is that its runtime can become prohibitively expensive as the dimension of \(X\) increases. In particular, for a target distribution \(P\) and proposal distribution \(\mu\) then the expected computational cost typically scales as
\(
\mathcal{O}\!\left(\,\operatorname*{ess\,sup}_x \{\left(\mathrm d P /\mathrm d\mu\right)(x)\}\,\right)\!.
\)
Unfortunately, this can grow \emph{exponentially} in the dimension even under a mild mismatch. For example, let \(P\) and \(\mu\) be \(d\)-dimensional product measures,
with \(P=\bigotimes_{i=1}^d P_i\) and
\(\mu=\bigotimes_{i=1}^d \mu_i\). In the i.i.d. case where
\(P_i=P_1\) and \(\mu_i=\mu_1\) for all \(i\)
\[
\frac{\mathrm d P}{\mathrm d\mu}(x_{1:d})=\prod_{i=1}^d \frac{\mathrm d P_1}{\mathrm d\mu_1}(x_i)
\quad\Rightarrow\quad
\operatorname*{ess\,sup}_{x_{1:d}}\frac{\mathrm d P}{\mathrm d\mu}(x_{1:d})
=\Big(\operatorname*{ess\,sup}_x \frac{\mathrm d P_1}{\mathrm d\mu_1}(x)\Big)^d = M^d,
\]
where \(M{=}\operatorname*{ess\,sup}_{x}(\mathrm dP_1/\mathrm d\mu_1)(x){>}1\) is the one-dimensional mismatch factor . This leads to an
\(M^d/C\) per-chain sample complexity in the grand-coupling setting. Thus, in the high-dimensional regimes typical of MCMC, a poorly matched proposal
can lead to prohibitively large runtimes.
\paragraph{Adaptive Proposal for Fast Coupling.}
To improve the runtime complexity, our starting point is the common assumption that sampling from each target distribution can be performed in \(\mathcal O(1)\) time. This assumption is particularly natural in the MCMC setting, where sampling from each transition kernel is typically cheap. Still, this does not automatically yield an efficient coupling: if the common measure \(\mu\) is chosen to be one of the target distributions, say \(P\), then sampling from another distribution \(Q\) may remain prohibitively expensive, since \(\operatorname*{ess\,sup}_x \{\left(\mathrm d Q/\mathrm dP\right)(x)\}\) can be large. Interestingly, this dimensional dependence can be replaced by dependence on the number of coupling proposals, which is often more manageable in practice, by setting \(\mu\) to the barycenter of the target measures. 
\par
Formally, consider $C$ probability measures $P^1,\dots,P^C$ such that sampling
from each $P^i$ can be done in $\mathcal{O}(1)$ time. 
Then, we choose the proposal to be the barycenter measure $\mu$ as their uniform mixture:
\begin{equation}
    \mu \triangleq \frac{1}{C} \sum_{j=1}^C P^j .
\end{equation}
The following lemma shows that this choice yields favorable runtime complexity.
\begin{lemma}[Mixture Barycenter Proposal]\label{lemma:barycenter_PML}
Let \(P^1,\dots,P^C\) be probability measures on \(\XSpace\), and the marked Poisson process defined previously with the common barycenter measure $\mu$ above. Then:
\begin{equation}
    P^i \ll \mu
    \qquad\text{and}\qquad
    \operatorname*{ess\,sup}_{x \in \XSpace }
    \left(\frac{\mathrm dP^i}{\mathrm d\mu}\right)(x)
    \le C,
    \qquad i=1,\dots,C .
\end{equation}
As a result, sampling from each measure has worst-case complexity \(O(C)\), and since they all share the same marked Poisson process, the expected number of samples per probability measure is \(O(1)\).
\end{lemma}
The proof is presented in Appendix \ref{proof:barycenter_PML}. Overall, this adaptive-proposal construction is computationally convenient. Since the common proposal measure is just a mixture of the individual kernels, the coupling step can be implemented using batched (parallel) operations: sampling mixture components, drawing proposals, and evaluating the corresponding density ratios.  Figure~\ref{fig:runtime_and_schematic} (Left) validates this scaling for $C=32$ over dimensions $d \in \{2^0,\ldots,2^9\}$. The targets are Gaussian distributions with different means and common identity covariance. We compare with a single-Gaussian proposal centered at the average target mean and covariance \(C I_d\), ensuring broad coverage. Our runtime remains nearly constant across dimensions, whereas the Gaussian-proposal baseline grows rapidly.

\section{Coupling Multiple Metropolis-Hastings Chains}\label{sec:mh_pml}
While list-level coupling provides useful guarantees and suggests a natural adaptive strategy, it is not trivial to implement for general continuous distributions. 
We therefore focus on adaptive Poisson matching constructions as tractable couplings for Markov chain Monte Carlo methods that use the Metropolis-Hastings (MH) algorithm.
To start, we briefly review the MH algorithm.
\subsection{Problem Setup}
\paragraph{Metropolis-Hastings (MH) Algorithm.} 
Given a target distribution \(\pi\), and a proposal kernel \(K\) with Lebesgue density \(k(\cdot\mid x)\), so that \(K(\mathrm dy\mid x)=k(y\mid x)\mathrm dy\), the Metropolis--Hastings algorithm proceeds
from a current state \(x\) by proposing \(y \sim k(\cdot\mid x)\) and accepting
the move with probability
\[
\alpha(x,y)
=
\min\left\{
1,
\frac{\pi(y)k(x\mid y)}{\pi(x)k(y\mid x)}
\right\}.
\]
If the proposal is accepted, the next state is \(y\); otherwise, the chain
remains at \(x\). This construction defines a Markov chain with a stationary distribution \(\pi\). Also see  Algorithm~\ref{alg:mh_kernel}, Appendix \ref{appendix:mh}.
\begin{remark}\label{rm:markov_kernel}
In the case of MH, we note that the resulting Markov transition kernel \(\kappa\) can be written as
\(
\kappa(\mathrm dy \mid x)
=
\alpha(x,y)\,k(y\mid x)\,\mathrm dy
+
\delta_x(\mathrm dy)\,r(x),
\)
where
\(
r(x)
=
1-\int \alpha(x,y)\,k(y\mid x)\,\mathrm dy
\)
is the rejection probability at state \(x\), and \(\delta_x\) denotes the Dirac measure at \(x\). Importantly, evaluating \(\kappa\) in closed form is nontrivial, since computing \(r(x)\) is generally intractable.
\end{remark}
The central goal in MCMC convergence analysis is to assess how close the chain's law at iteration \(t\), denoted by \(\pi_t\), is to the target distribution \(\pi\). 
To this end, coupling provides a principled way to study convergence by jointly evolving two or more Markov chains. Following  \cite{johnson1998coupling}, we consider a \emph{grand coupling} of \(C\) chains, all initialized from a prior distribution \(\pi_0\) and evolved using a coupled kernel.
\begin{definition}[Grand coupling]\label{def:grand_coupling}
Let \(\kappa\) be a Markov kernel on a measurable space \(\XSpace\).
A \emph{grand coupling} of \(\kappa\) for \(C\) chains is a Markov 
kernel \(\bar{\kappa}\) on the product space \(\XSpace^C\) such that, if
\[
(X_{t+1}^{(1)},\dots,X_{t+1}^{(C)})
\sim
\bar{\kappa}(\cdot \mid (x_t^{(1)},\dots,x_t^{(C)})),
\]
then
\[
X_{t+1}^{(i)} \sim \kappa(\cdot \mid x_t^{(i)}),
\qquad i=1,\dots,C.
\]
The associated grand meeting time is
\[
\tau \triangleq \inf\{t \ge 0 : X_t^{(1)} = \cdots = X_t^{(C)}\}.
\]
We further assume faithfulness: once two or more chains meet at some time
\(\tilde{\tau}\), they are updated identically thereafter and remain equal for
all \(t \ge \tilde{\tau}\).
\end{definition}%
Following \cite{biswas2019estimating,johnson1996studying}, faster coalescence is desirable because the tail probability \(\Pr(\tau>t)\) often controls convergence bounds, so reducing this term leads to tighter bounds. See Appendix \ref{app:futher_exps} for further details.
\subsection{Poisson Matching
Coupled Kernel}\label{sec:mcmc_coupling}
\begin{algorithm}[t]
\SetAlgoLined
\DontPrintSemicolon
\LinesNotNumbered
\SetKwInOut{Input}{Input}
\SetKwInOut{Output}{Output}

\Input{Current states $x_t^{(1)},\dots,x_t^{(C)}$}
\Output{Updated states $x_{t+1}^{(1)},\dots,x_{t+1}^{(C)}$}

\makebox[0.28\linewidth][l]{Set lifted target measures:}
$\displaystyle
P_i'(\mathrm dy,\mathrm du)
=
k(y \mid x_t^{(i)})\,\mathrm dy\,
\mathrm{Ber}\!\bigl(\alpha(x_t^{(i)},y)\bigr)(\mathrm du),
\quad i=1,\dots,C
$\;

\makebox[0.28\linewidth][l]{Proposal measure:}
$\displaystyle
\mu(\mathrm dy,\mathrm du)
=
\frac{1}{C}\sum_{i=1}^C P_i'(\mathrm dy,\mathrm du)
$\;

\makebox[0.28\linewidth][l]{Shared Poisson process:}
Generate $\Omega=\{(S_j,Y_j,U_j)\}_{j\ge 1}$ with reference measure $\dd s \otimes \mu$\;

\For{$i=1,\dots,C$}{
    \makebox[0.28\linewidth][l]{PML selection:}
    $\displaystyle
    J^{(i)}
    =
    \arg\min_{j\ge 1}
    \left\{
    S_j
    \left(
    \frac{\mathrm d P_i'}{\mathrm d\mu}(Y_j,U_j)
    \right)^{-1}
    \right\}
    $\;

    \makebox[0.28\linewidth][l]{MH update:}
    $\displaystyle
    x_{t+1}^{(i)}
    \gets
    U_{J^{(i)}}Y_{J^{(i)}}
    +
    \bigl(1-U_{J^{(i)}}\bigr)x_t^{(i)}
    $\;
}

\caption{Metropolis--Hastings Coupled Kernel with Poisson Matching}
\label{alg:lifted_mh_coupling}
\end{algorithm}
We now introduce the \emph{joint coupled kernel} method, which couples directly at the level of the full Markov transition kernel \(\kappa\). Since \(\kappa\) is typically unavailable in closed form for scoring in~\eqref{PP_opt} (see Remark~\ref{rm:markov_kernel}), we instead lift the construction to the augmented proposal--acceptance space. For each chain \(i\), define the lifted measure
\begin{equation}
P_i'(\mathrm dy,\mathrm du)
=
k(y \mid x_t^{(i)})\,\mathrm dy\,
\mathrm{Ber}\!\bigl(\alpha(x_t^{(i)},y)\bigr)(\mathrm du),
\end{equation}
where \(Y_t^{(i)} \sim k(\cdot \mid x_t^{(i)})\), and
\(U_t^{(i)} \mid \{Y_t^{(i)}=y\} \sim \mathrm{Ber}(\alpha(x_t^{(i)},y))\).
The chains then update their respective next state by
\begin{equation}
x_{t+1}^{(i)}
=
\mathbf{1}\{U_t^{(i)}=1\}Y_t^{(i)}
+
\mathbf{1}\{U_t^{(i)}=0\}x_t^{(i)}.
\label{single_eq}
\end{equation}
This lifted representation avoids treating the rejection component as an explicit atom inside the PML construction; instead, rejection is encoded by the auxiliary variable \(U_t^{(i)}\), while the pushforward through \eqref{single_eq} recovers the original MH transition kernel. How does this augmented-space strategy compare with the non-augmented one, assuming full knowledge about the point mass measure? Surprisingly, Lemma \ref{lem:match_prob} shows that the augmented construction yields a better matching probability than the direct coupling scheme based on $\kappa$; see Appendix \ref{aug_proof} for the proof.
\begin{lemma}\label{lem:match_prob}
Consider two chains at states $x^1_t$ and $x^2_t$ respectively. At $t+1$, let $\mathbb{P}_{\mathrm{aug}}(x^1_t,x^2_t)$ denote the meeting probability under the product-augmented construction, and let $\mathbb{P}_{\mathrm{direct}}(x^1_t,x^2_t)$ denote the meeting probability under the original direct construction. Then
\[
\mathbb{P}_{\mathrm{aug}}(x^1_t,x^2_t) \geq \mathbb{P}_{\mathrm{direct}}(x^1_t,x^2_t).
\]
\end{lemma}
To construct a shared source of common randomness, we use a marked Poisson point process
\(\Omega=\{(Y_j,U_j,S_j)\}_{j\ge1}\),\footnote{We use \(j\) to index the Poisson points and \(t\) to denote the Markov chain time index.} whose marks \((Y_j,U_j)\) are drawn i.i.d. from the mixture measure\footnote{For simple implementation, one may also use
\(\mu(\mathrm dy,\mathrm du)
\triangleq
\mathrm{Ber}\!\bigl(1/2\bigr)(\mathrm du)
\times
\frac{1}{C}\sum_{i=1}^C k(y\mid x_t^{(i)})\,\mathrm dy,
\)
with upper-bound ratio and expected runtime complexity \(2C\).}
\[
\mu(\mathrm dy,\mathrm du)
\triangleq
\frac{1}{C}\sum_{i=1}^C
k(y\mid x_t^{(i)})\,\mathrm dy\,
\mathrm{Ber}\!\bigl(\alpha(x_t^{(i)},y)\bigr)(\mathrm du).
\]
We then apply the PMC selection step using \(P'_i\) as the target measure and \(\mu\) as the common proposal. The procedure is summarized in Algorithm~\ref{alg:lifted_mh_coupling}. Finally, for completeness, we also outline the \emph{two-stage coupled kernel}, which decomposes the procedure into proposal coupling followed by coupling of the MH accept/reject decisions. 
Compared with this two-stage approach, the joint construction introduced here provides stronger matching-probability guarantees, detailed in Theorem \ref{thm:compare_bound} in Appendix \ref{appendix:mh_matching_prob_1step}, a behavior we also observe consistently in our experiments. 

\section{Experiments}\label{sec:experiments}
\paragraph{Multi-Marginal Coupling.}
We begin with the finite-state setup from Section~\ref{sec:multi-marginal},
shown in the left panel of  Figure~\ref{fig:chain_comparison}. 
We evaluate the proposed multi-marginal coupling methods on randomly generated sparse discrete distributions over a state space with $|\mathcal X|=60$, where each distribution is supported on $5$ states, which induces an irregular overlap structure across pairs of measures. 
We estimate $\mathbb{E}[G]$ for $C \in \{2^1,\ldots,2^5\}$ over $2 \times 10^5$ runs, ensuring that multiple measures share at least one state. Our baselines are two maximal pairwise-coupling schemes. The first, called ``random anchor'', selects a reference measure at random in each trial and couples all other measures to it. The second, random sequence, applies pairwise coupling along a
randomly generated sequence of pairs. 
\begin{figure}[t]
\centering

\begin{minipage}[t]{0.48\linewidth}
\centering
\vspace{-130pt}
\scriptsize

\begin{tikzpicture}
\matrix[
    matrix of nodes,
    row sep=2pt,
    column sep=3pt,
    nodes={anchor=west, font=\scriptsize}
] {
    \tikz{\draw[blue, thick] (0,0) -- (0.35,0); \fill[blue] (0.175,0) circle (1.6pt);}
    & Random Anchor
    &
    \tikz{\draw[red, thick] (0,0) -- (0.35,0); \filldraw[red] (0.175,-0.04) rectangle +(0.08,0.08);}
    & Random Sequence
    \\
    \tikz{\draw[orange, thick] (0,0) -- (0.35,0); \filldraw[orange] (0.175,0.06) -- (0.105,-0.06) -- (0.245,-0.06) -- cycle;}
    & List Coupling (Ours)
    &
    \tikz{\draw[green!60!black, thick] (0,0) -- (0.35,0); \filldraw[green!60!black] (0.175,0)
      (0.175,0.08) -- (0.255,0) -- (0.175,-0.08) -- (0.095,0) -- cycle;}
    & Poisson Matching (Ours)
    \\
    \tikz{\draw[black, thick, dashed] (0,0) -- (0.35,0);
    \node[black] at (0.175,0) {\scriptsize$\star$};}
    & Lower Bound
    & & \\
};
\end{tikzpicture}

\vspace{1pt}

\renewcommand{\arraystretch}{1.15}
\begin{tabular}{lccccc}
\toprule
 & \multicolumn{5}{c}{Number of Measures $C$} \\
\cmidrule(lr){2-6}
Method & $2$ & $4$ & $8$ & $16$ & $32$ \\
\midrule
Poisson Matching        & 1.83 & 3.61 & 7.10 & 11.97 & \textbf{18.99} \\
List          & \textbf{1.71} & \textbf{3.44} & \textbf{6.82} & \textbf{11.56} & 19.28 \\
Random Anchor & \textbf{1.71} & 3.62 & 7.37 & 13.56 & 23.61 \\
Random Sequence  & \textbf{1.71} & 3.63 & 7.38 & 13.59 & 23.82 \\
\midrule
Lower Bound   & 1.71 & 3.40 & 6.60 & 10.28 & 14.60 \\
\bottomrule
\end{tabular}
\end{minipage}
\hspace{10pt}
\begin{minipage}[t]{0.48\linewidth}
\centering
\begin{tikzpicture}
\begin{axis}[
    width=\linewidth,
    height=0.75\linewidth,
    xlabel={Number of distributions $C$},
    ylabel={Estimated $\mathbb{E}[G]$},
    xlabel style={font=\footnotesize},
    ylabel style={font=\footnotesize},
    tick label style={font=\footnotesize},
    xmin=2, xmax=32,
    ymin=0, ymax=32,
    xtick={2,4,8,16,32},
    ymajorgrids=true,
    xmajorgrids=true,
    grid style={dashed, gray!30},
    mark size=2.2pt,
]

\addplot+[blue, mark=o, thick] coordinates {
    (2,1.64) (4,3.29) (8,7.04) (16,14.96) (32,30.89)
};

\addplot+[red, mark=square, thick] coordinates {
    (2,1.64) (4,3.20) (8,6.90) (16,14.80) (32,30.85)
};

\addplot+[orange, mark=triangle, thick] coordinates {
    (2,1.64) (4,3.14) (8,6.11) (16,12.03) (32,23.93)
};

\addplot+[green!60!black, mark=diamond, thick] coordinates {
    (2,1.64) (4,2.87) (8,5.43) (16,10.46) (32,20.57)
};

\addplot+[black, mark=star, thick, dashed] coordinates {
    (2,1.6321205588285577)
    (4,2.896361676485673)
    (8,5.424843911799904)
    (16,10.481808382428365)
    (32,20.595737323685288)
};

\end{axis}
\end{tikzpicture}
\end{minipage}

\caption{Estimated $\mathbb E[G]$ across different coupling strategies across number of measures $C$. Left: random discrete measures. Right: exponential target measures ( scale$=1$ but different locations).}
\label{fig:chain_comparison}
\end{figure}
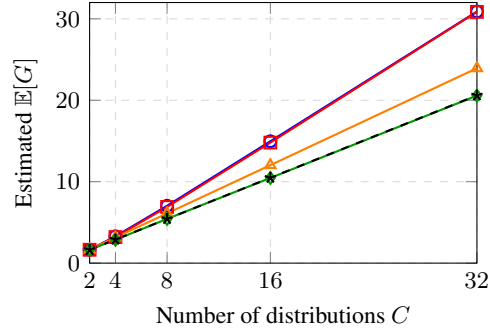%
\par
For $C=2$, the pairwise baselines and the list-level
coupling recover the optimal coupling performance. For larger $C$, our schemes consistently achieve substantially lower $\mathbb{E}[G]$ than the pairwise baselines. 
The list-level coupling outperforms Poisson matching for $C \in \{2,4,8,16\}$, likely because the sparse overlap structure resembles the
optimality-achieving setting in Example~\ref{example:optimal_list}. 
In such settings, list-level coupling can exploit local overlap groups, whereas Poisson matching does not generally maximize all pairwise matches simultaneously. 
As $C$ grows, the combined supports cover more of the state space, making this local-overlap advantage less pronounced; in this regime, Poisson matching becomes more effective and can outperform the list-level coupling.
\par
The right panel of Figure~\ref{fig:chain_comparison} depicts a continuous setting where the recursive list-level step is analytically computable, using the shifted exponential family \(P_i=i+\mathrm{Exp}(1)\), \(i=0,\ldots,C-1\). 
The list-level coupling again outperforms the pairwise baselines. 
Interestingly, Poisson matching exactly matches the lower bound; it turns out that theoretically, Poisson matching achieves an optimal \(\mathbf G^*\) in this particular setting. The proof is given in Appendix~\ref{example:pml_shifted_exp}, where we show that $\mathbf{G}^* = C-(C-1)e^{-1}.$
\paragraph{Coupling Metropolis--Hastings chains.}
We evaluate our proposed scheme on two representative random-walk
Metropolis--Hastings (RWMH) kernels in Figure~\ref{fig:coupling-dimension-combined}.
The first row considers a Gaussian target
\(\pi=\mathcal N(0,\mathrm I_d)\), with chains initialized from
\(\pi_0=\mathcal N(\mathbf 1_d,16\mathrm I_d)\). The second row considers a
heavy-tailed Cauchy target sampled using a Student-\(t\) random-walk proposal
with \(2\) degrees of freedom, where the chains start with $\pi_0=\mathcal{N}(0, \mathrm I_d)$. We set the proposal scale to
\(\sigma{=}2.4/\sqrt d\), following the classical high-dimensional optimal-scaling
heuristic for random-walk Metropolis~\cite{gelman1997weak}, which we use here as
a practical tuning rule. 

We compare our proposed scheme with grand-coupling extensions of the pairwise
methods from \citep{johnson1998coupling} and \citep{wang2021maximal}, namely
Algorithms~1 and~3 of \cite{wang2021maximal} for general-purpose kernels. Since
these methods were originally developed for coupling two chains, we adapt them
to the \(C\)-chain setting by selecting one chain as a reference and coupling
each of the remaining \(C-1\) chains to it via maximal coupling; we refer to this
as star coupling. The resulting baselines can be viewed as the two-step and
one-step counterparts of our approach, respectively. Note that, in contrast to the multi-marginal setting, we find that using a fixed anchor, or reference chain, across iterations works best for these baselines. For each setting, we measure the average meeting time as a function of the
number of chains \(C\) with fixed \(d\), and as a function of dimension \(d\)
with fixed \(C\). Across both targets, our method achieves smaller meeting times
than the pairwise star-coupling baselines. We attribute the improvement to the adaptive construction of the shared point process, which couples all chains jointly rather than forcing every match through a fixed
reference chain.
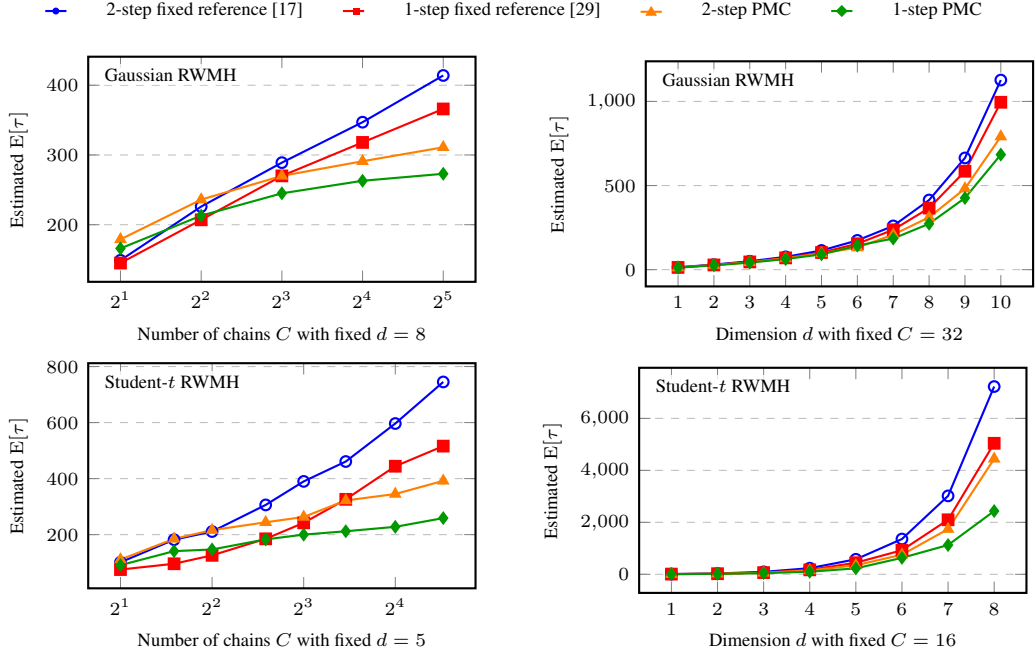
\begin{figure}[t]
\centering

\begin{tikzpicture}
\matrix[
    matrix of nodes,
    row sep=2pt,
    column sep=8pt,
    nodes={anchor=west, font=\scriptsize}
] {
    \tikz{\draw[blue, thick] (0,0) -- (0.28,0); \fill[blue] (0.14,0) circle (1.3pt);}
    & 2-step fixed reference \cite{johnson1998coupling} 
    &
    \tikz{\draw[red, thick] (0,0) -- (0.28,0); \fill[red] (0.14,-0.045) rectangle +(0.09,0.09);}
    & 1-step fixed reference \cite{wang2021maximal}
    &
    \tikz{\draw[orange, thick] (0,0) -- (0.28,0); \filldraw[orange] (0.14,0.065) -- (0.075,-0.045) -- (0.205,-0.045) -- cycle;}
    & 2-step PMC
    &
    \tikz{\draw[green!60!black, thick] (0,0) -- (0.28,0); \filldraw[green!60!black] (0.14,0.07) -- (0.21,0) -- (0.14,-0.07) -- (0.07,0) -- cycle;}
    & 1-step PMC \\
};
\end{tikzpicture}

\vspace{0.6em}

\begin{minipage}[t]{0.48\textwidth}
\centering
\begin{tikzpicture}
\begin{axis}[
    width=\textwidth,
    height=0.68\textwidth,
    xlabel={Number of chains $C$ with fixed $d=8$},
    ylabel={Estimated $\mathrm{E}[\tau]$},
    label style={font=\scriptsize},
    tick label style={font=\scriptsize},
    xtick={2,4,8,16,32},
    ymajorgrids=true,
    grid style=dashed,
    xmode=log,
    log basis x={2},
    thick,
    mark options={scale=1.0}
]

\addplot[color=blue, mark=o]
coordinates {
    (2,149) (4,226) (8,289) (16,347) (32,414)
};

\addplot[color=red, mark=square*]
coordinates {
    (2,145) (4,207) (8,270) (16,318) (32,366)
};

\addplot[color=orange, mark=triangle*]
coordinates {
    (2,179) (4,236) (8,270) (16,291) (32,311)
};

\addplot[color=green!60!black, mark=diamond*]
coordinates {
    (2,166) (4,213) (8,245) (16,263) (32,273)
};

\node[
    anchor=north west,
    font=\scriptsize,
    fill=white,
    fill opacity=0.8,
    text opacity=1,
    inner sep=1.5pt
] at (rel axis cs:0.03,0.97) {Gaussian RWMH};

\end{axis}
\end{tikzpicture}
\end{minipage}
\hfill
\begin{minipage}[t]{0.48\textwidth}
\centering
\begin{tikzpicture}
\begin{axis}[
    width=\textwidth,
    height=0.68\textwidth,
    xlabel={Dimension $d$ with fixed $C=32$},
    ylabel={Estimated $\mathrm{E}[\tau]$},
    label style={font=\scriptsize},
    tick label style={font=\scriptsize},
    xtick={1,2,3,4,5,6,7,8,9,10},
    ymajorgrids=true,
    grid style=dashed,
    thick,
    mark options={scale=1.0}
]

\addplot[color=blue, mark=o]
coordinates {
    (1,14.1) (2,30) (3,50) (4,76) (5,114)
    (6,174) (7,260) (8,414) (9,664) (10,1127)
};

\addplot[color=red, mark=square*]
coordinates {
    (1,13.3) (2,28) (3,46) (4,70) (5,102)
    (6,154) (7,237) (8,366) (9,585) (10,995)
};

\addplot[color=orange, mark=triangle*]
coordinates {
    (1,12.17) (2,27) (3,44) (4,67) (5,97)
    (6,131) (7,208) (8,311) (9,482) (10,791)
};

\addplot[color=green!60!black, mark=diamond*]
coordinates {
    (1,11.9) (2,26) (3,42) (4,63) (5,91)
    (6,143) (7,185) (8,273) (9,426) (10,684)
};

\node[
    anchor=north west,
    font=\scriptsize,
    fill=white,
    fill opacity=0.8,
    text opacity=1,
    inner sep=1.5pt
] at (rel axis cs:0.03,0.97) {Gaussian RWMH};

\end{axis}
\end{tikzpicture}
\end{minipage}


\begin{minipage}[t]{0.48\textwidth}
\centering
\begin{tikzpicture}
\begin{axis}[
    width=\textwidth,
    height=0.68\textwidth,
    xlabel={Number of chains $C$ with fixed $d=5$},
    ylabel={Estimated $\mathrm{E}[\tau]$},
    label style={font=\scriptsize},
    tick label style={font=\scriptsize},
    xtick={2,4,8,16},
    ymajorgrids=true,
    grid style=dashed,
    xmode=log,
    log basis x={2},
    thick,
    mark options={scale=1.0}
]

\addplot[color=blue, mark=o]
coordinates {
    (2,101.2) (3,182) (4,211) (6,306) (8,390) (11,461) (16,596.6) (23,745)
};

\addplot[color=red, mark=square*]
coordinates {
    (2,75.3) (3,96) (4,126) (6,185) (8,242) (11,326) (16,444.14) (23,516)
};

\addplot[color=orange, mark=triangle*]
coordinates {
    (2,111) (3,185) (4,216) (6,244) (8,263) (11,322) (16,345) (23,392)
};

\addplot[color=green!60!black, mark=diamond*]
coordinates {
    (2,91) (3,141) (4,147) (6,183) (8,200) (11,212) (16,227.8) (23,259)
};

\node[
    anchor=north west,
    font=\scriptsize,
    fill=white,
    fill opacity=0.8,
    text opacity=1,
    inner sep=1.5pt
] at (rel axis cs:0.03,0.97) {Student-$t$ RWMH};

\end{axis}
\end{tikzpicture}
\end{minipage}
\hfill
\begin{minipage}[t]{0.48\textwidth}
\centering
\begin{tikzpicture}
\begin{axis}[
    width=\textwidth,
    height=0.68\textwidth,
    xlabel={Dimension $d$ with fixed $C=16$},
    ylabel={Estimated $\mathrm{E}[\tau]$},
    label style={font=\scriptsize},
    tick label style={font=\scriptsize},
    xtick={1,2,3,4,5,6,7,8},
    ymajorgrids=true,
    grid style=dashed,
    thick,
    mark options={scale=1.0}
]

\addplot[color=blue, mark=o]
coordinates {
    (1,9.4) (2,31.3) (3,93.3) (4,235.2)
    (5,576.6) (6,1356.6) (7,3017.3) (8,7226)
};

\addplot[color=red, mark=square*]
coordinates {
    (1,7.8) (2,22.6) (3,64.9) (4,170.9)
    (5,444.14) (6,925.6) (7,2097.9) (8,5042)
};

\addplot[color=orange, mark=triangle*]
coordinates {
    (1,6.5) (2,21.8) (3,56.9) (4,148.2)
    (5,345) (6,754.5) (7,1738) (8,4443)
};

\addplot[color=green!60!black, mark=diamond*]
coordinates {
    (1,5.4) (2,16.9) (3,44.9) (4,98.4)
    (5,227.8) (6,632.7) (7,1125.9) (8,2434)
};

\node[
    anchor=north west,
    font=\scriptsize,
    fill=white,
    fill opacity=0.8,
    text opacity=1,
    inner sep=1.5pt
] at (rel axis cs:0.03,0.97) {Student-$t$ RWMH};

\end{axis}
\end{tikzpicture}
\end{minipage}
\vspace{-5pt}
\caption{Comparison of estimated meeting times across Gaussian and Student-$t$ targets. 
Top row: Gaussian target; bottom row: Student-$t$ target. 
Left column: meeting time versus the number of coupled chains; right column: meeting time versus dimension. 
Results are averaged over 10000 runs.}
\vspace{-10pt}
\label{fig:coupling-dimension-combined}
\end{figure}

\paragraph{Additional experiments.}
Finally, we provide further discussion of other kernel and additional
details on grand couplings for convergence diagnostics in
Appendix~\ref{app:futher_exps}. There, we also present an alternative bound different from that of \citet{johnson1996studying}, targetting the case where rejection
coefficient between the two measures is not available. We also discuss applications
to other diagnostic methods, such as weight harmonization
\citep{corenflos2025coupling}.
\section{Related Work}
\paragraph{Multi-marginal coupling.}
Our construction builds on the Poisson matching lemma (PML)
\citep{li2021unified}, which studies matching from shared Poisson randomness
primarily for information-theoretic achievability using Poisson Monte Carlo
(PMC), also known as \(A^*\) sampling
\citep{maddison2014sampling,theis2022algorithms,flamich2022fast,flamich2024data}. This framework has recently gained traction in compression
\citep{flamich2022fast,flamich2023greedy,rowan2026one}, alongside related
schemes based on different sampling procedures
\citep{phanchannel,daliri2025coupling,phan2024importance}. Several recent works study list-level coupling questions motivated by multi-draft speculative decoding \citep{rowan2025list,khisti2024multi}; however, prior to our work, optimal achievability bounds for the general multi-marginal formulation were not known. For couplings of multiple measures, to the best of our knowledge, the closest
existing work is that of \citet{angel2019pairwise}, who study simultaneous
pairwise disagreement guarantees for discrete distributions. In contrast, we
study multi-marginal coupling through a cluster-count objective, relate it to
list-level and pairwise mismatch, and develop efficient constructions that
extend naturally to continuous measures.

\paragraph{Markov chain coupling.}
Couplings are widely used to study the convergence of Markov chains and to design
MCMC convergence diagnostics
\citep{johnson1996studying,biswas2019estimating,jacob2020unbiased,grimmett2025coalescence}.
Existing approaches include contraction-based couplings, such as reflection
couplings \citep{eberle2016reflection}, and exact-meeting constructions based
on maximal couplings of Metropolis--Hastings kernels \citep{wang2021maximal} or
Gibbs kernels \citep{trippe2021optimal}. Our work instead focuses on
multi-chain couplings that directly target list-level and pairwise matching
events through a multi-marginal formulation, and can be viewed as extending
maximal-coupling ideas for Metropolis--Hastings kernels to the multi-chain
setting.
\section{Conclusion}
We studied coupling constructions for Markov chains from a multi-marginal
perspective, with the goal of encouraging matching behavior beyond the classical
pairwise maximal-coupling setting. Our construction leads to an optimization
formulation for multi-marginal coupling, for which we provide a general bounds as well as optimal
solutions in several specific settings. These solutions are attained using either
list-level coupling or Poisson matching constructions. From a theoretical
perspective, it would be interesting to develop more general-purpose solutions
to this optimization problem and to design efficient algorithms for computing or
approximating them. Such developments could, in turn, lead to stronger and more broadly applicable
coupling constructions for Markov chains.

Related to convergence analysis, a promising direction for future work is to build and derive sharper convergence diagnostics based on multiple coupled chains, potentially enabling more informative assessments of mixing and convergence than standard two-chain methods. Another important application of coupled Markov chains is perfect simulation,
such as coupling from the past~\citep{huber2016perfect}. These methods can be
viewed as an extreme form of convergence diagnosis: they couple chains from all
possible initial states and check whether they have coalesced, yielding exact
samples when applicable. In contrast, our methods track coalescence among only
finitely many chains as a practical diagnostic. A natural future direction is to
explore whether Poisson coupling techniques can be used to construct update
functions for perfect simulation as well.

\section*{Acknowledgment}
Buu Phan and Ashish Khisti were supported by the NSERC Discovery Grant.
The authors acknowledge financial support from Imperial College London through an Imperial College Research Fellowship grant awarded to Gergely Flamich. Shahab Asoodeh was supported
by the NSERC Discovery Grant.

Resources used in preparing this research were provided, in part, by the Province of Ontario,
the Government of Canada through CIFAR, and companies sponsoring the Vector Institute
\href{www.vectorinstitute.ai/partnerships/}{www.vectorinstitute.ai/partnerships/}.

\bibliographystyle{plainnat} 
\bibliography{references} 

@article{maddison2014sampling,
  title={A* Sampling},
  author={Maddison, Chris J and Tarlow, Daniel and Minka, Tom},
  journal={Advances in Neural Information Processing Systems},
  volume={27},
  pages={3086--3094},
  year={2014}
}

@article{flamich2022fast,
  title={Fast relative entropy coding with a* coding},
  author={Flamich, Gergely and Markou, Stratis and Hern{\'a}ndez-Lobato, Jos{\'e} Miguel},
  journal={arXiv preprint arXiv:2201.12857},
  year={2022}
}

@book{levin2017markov,
  title={Markov chains and mixing times},
  author={Levin, David A and Peres, Yuval},
  volume={107},
  year={2017},
  publisher={American Mathematical Soc.}
}

@phdthesis{flamich2024data,
  title = {Data Compression with Relative Entropy Coding},
  url = {https://www.repository.cam.ac.uk/handle/1810/385303},
  doi = {10.17863/CAM.118972},
  school = {Apollo - University of Cambridge Repository},
  author = {Flamich, Gergely},
  year = {2024},
  note = {PhD Thesis},
}

@article{li2021unified,
  title={A unified framework for one-shot achievability via the Poisson matching lemma},
  author={Li, Cheuk Ting and Anantharam, Venkat},
  journal={IEEE Transactions on Information Theory},
  volume={67},
  number={5},
  pages={2624--2651},
  year={2021},
  publisher={IEEE}
}

@inproceedings{wang2021maximal,
  title={Maximal couplings of the Metropolis-Hastings algorithm},
  author={Wang, Guanyang and O’Leary, John and Jacob, Pierre},
  booktitle={International Conference on Artificial Intelligence and Statistics},
  pages={1225--1233},
  year={2021},
  organization={PMLR}
}

@article{biswas2019estimating,
  title={Estimating convergence of Markov chains with L-lag couplings},
  author={Biswas, Niloy and Jacob, Pierre E and Vanetti, Paul},
  journal={Advances in neural information processing systems},
  volume={32},
  year={2019}
}

@article{grimmett2025coalescence,
  title={Coalescence in Markov chains},
  author={Grimmett, Geoffrey R and Holmes, Mark},
  journal={arXiv preprint arXiv:2510.13572},
  year={2025}
}

@article{angel2019pairwise,
  title={Pairwise optimal coupling of multiple random variables},
  author={Angel, Omer and Spinka, Yinon},
  journal={arXiv preprint arXiv:1903.00632},
  year={2019}
}

@article{brooks1998general,
  title={General methods for monitoring convergence of iterative simulations},
  author={Brooks, Stephen P and Gelman, Andrew},
  journal={Journal of computational and graphical statistics},
  volume={7},
  number={4},
  pages={434--455},
  year={1998},
  publisher={Taylor \& Francis}
}

@article{gelman1992inference,
  title={Inference from iterative simulation using multiple sequences},
  author={Gelman, Andrew and Rubin, Donald B},
  journal={Statistical science},
  volume={7},
  number={4},
  pages={457--472},
  year={1992},
  publisher={Institute of Mathematical Statistics}
}

@article{corenflos2025coupling,
  title={A coupling-based approach to f-divergences diagnostics for Markov chain Monte Carlo},
  author={Corenflos, Adrien and Dau, Hai-Dang},
  journal={arXiv preprint arXiv:2510.07559},
  year={2025}
}

@article{johnson1996studying,
  title={Studying convergence of Markov chain Monte Carlo algorithms using coupled sample paths},
  author={Johnson, Valen E},
  journal={Journal of the American Statistical Association},
  volume={91},
  number={433},
  pages={154--166},
  year={1996},
  publisher={Taylor \& Francis}
}

@article{johnson1998coupling,
  title={A coupling-regeneration scheme for diagnosing convergence in Markov chain Monte Carlo algorithms},
  author={Johnson, Valen E},
  journal={Journal of the American Statistical Association},
  volume={93},
  number={441},
  pages={238--248},
  year={1998},
  publisher={Taylor \& Francis}
}

@inproceedings{daliri2025coupling,
  title={Coupling without communication and drafter-invariant speculative decoding},
  author={Daliri, Majid and Musco, Christopher and Suresh, Ananda Theertha},
  booktitle={2025 IEEE International Symposium on Information Theory (ISIT)},
  pages={1--6},
  year={2025},
  organization={IEEE}
}

@article{gelman1997weak,
  title={Weak convergence and optimal scaling of random walk Metropolis algorithms},
  author={Gelman, Andrew and Gilks, Walter R and Roberts, Gareth O},
  journal={The annals of applied probability},
  volume={7},
  number={1},
  pages={110--120},
  year={1997},
  publisher={Institute of Mathematical Statistics}
}

@article{girolami2011riemann,
  title={Riemann manifold langevin and hamiltonian monte carlo methods},
  author={Girolami, Mark and Calderhead, Ben},
  journal={Journal of the Royal Statistical Society Series B: Statistical Methodology},
  volume={73},
  number={2},
  pages={123--214},
  year={2011},
  publisher={Oxford University Press}
}

@article{sason2016f,
  title={$ f $-divergence Inequalities},
  author={Sason, Igal and Verd{\'u}, Sergio},
  journal={IEEE Transactions on Information Theory},
  volume={62},
  number={11},
  pages={5973--6006},
  year={2016},
  publisher={IEEE}
}

@inproceedings{theis2022algorithms,
  title={Algorithms for the communication of samples},
  author={Theis, Lucas and Ahmed, Noureldin Y},
  booktitle={International Conference on Machine Learning},
  pages={21308--21328},
  year={2022},
  organization={PMLR}
}

@article{huber2016perfect,
  title={Perfect simulation},
  author={Huber, Mark L},
  journal={Monographs on Statistics and Applied Probability},
  volume={148},
  pages={148},
  year={2016}
}

@inproceedings{phan2024importance,
  title={Importance matching lemma for lossy compression with side information},
  author={Phan, Buu and Khisti, Ashish and Louizos, Christos},
  booktitle={International Conference on Artificial Intelligence and Statistics},
  pages={1387--1395},
  year={2024},
  organization={PMLR}
}

@inproceedings{phanchannel,
  title={Channel Simulation and Distributed Compression with Ensemble Rejection Sampling},
  author={Phan, Buu and Khisti, Ashish J},
  year={2025},
  booktitle={The Thirty-ninth Annual Conference on Neural Information Processing Systems}
}

@article{flamich2023greedy,
  title={Greedy Poisson rejection sampling},
  author={Flamich, Gergely},
  journal={Advances in Neural Information Processing Systems},
  volume={36},
  pages={37089--37127},
  year={2023}
}

@article{rowan2026one,
  title={One-Shot Broadcast Joint Source-Channel Coding with Codebook Diversity},
  author={Rowan, Joseph and Phan, Buu and Khisti, Ashish},
  journal={arXiv preprint arXiv:2601.10648},
  year={2026}
}

@article{eberle2016reflection,
  title={Reflection couplings and contraction rates for diffusions},
  author={Eberle, Andreas},
  journal={Probability theory and related fields},
  volume={166},
  number={3},
  pages={851--886},
  year={2016},
  publisher={Springer}
}

@article{jacob2020unbiased,
  title={Unbiased Markov chain Monte Carlo methods with couplings},
  author={Jacob, Pierre E and O’leary, John and Atchad{\'e}, Yves F},
  journal={Journal of the Royal Statistical Society Series B: Statistical Methodology},
  volume={82},
  number={3},
  pages={543--600},
  year={2020},
  publisher={Oxford University Press}
}

@article{trippe2021optimal,
  title={Optimal transport couplings of Gibbs samplers on partitions for unbiased estimation},
  author={Trippe, Brian L and Nguyen, Tin D and Broderick, Tamara},
  journal={arXiv preprint arXiv:2104.04514},
  year={2021}
}

@article{rowan2025list,
  title={List-Level Distribution Coupling with Applications to Speculative Decoding and Lossy Compression},
  author={Rowan, Joseph and Phan, Buu and Khisti, Ashish},
  journal={arXiv preprint arXiv:2506.05632},
  year={2025}
}

@article{khisti2024multi,
  title={Multi-draft speculative sampling: Canonical decomposition and theoretical limits},
  author={Khisti, Ashish and Ebrahimi, M Reza and Dbouk, Hassan and Behboodi, Arash and Memisevic, Roland and Louizos, Christos},
  journal={arXiv preprint arXiv:2410.18234},
  year={2024}
}


\newpage
\appendix
\section{Auxiliary Algorithms}\label{appendix:aux_algo}
\subsection{Maximal Coupling}\label{appendix:max_coupling}

\paragraph{Achievability.}
Let \(P\) and \(Q\) be probability measures on a measurable space
\((\XSpace,\mathcal X)\). Their meet measure is defined by
\[
\gamma := P\wedge Q,
\]
that is, \(\gamma\) is the largest measure dominated by both \(P\) and \(Q\).
Equivalently, when \(P\) and \(Q\) have densities \(p\) and \(q\) with respect
to the same base measure, the common component has density
\[
\gamma(x)=\min\{p(x),q(x)\},
\]
with the abuse of notation, we denote the density with the same letter when no confusion arises. The total mass of the common component satisfies
\[
\gamma(\XSpace)
=
1-d_{\mathrm{TV}}(P,Q),
\]
where
\[
d_{\mathrm{TV}}(P,Q)
=
\sup_{A\in\mathcal X}|P(A)-Q(A)|.
\]

A maximal coupling can be constructed by first identifying this common
component. Namely,
\begin{align*}
\text{(i)} &\quad
\text{with probability } \gamma(\XSpace),
\text{ sample } X^P=X^Q \sim \gamma/\gamma(\XSpace),\\
\text{(ii)} &\quad
\text{otherwise, sample }
X^P\sim \frac{P-\gamma}{d_{\mathrm{TV}}(P,Q)}
\text{ and }
X^Q\sim \frac{Q-\gamma}{d_{\mathrm{TV}}(P,Q)}.
\end{align*}
By construction, \(X^P\sim P\) and \(X^Q\sim Q\). Moreover,
\[
\Pr(X^P=X^Q)
=
\gamma(\XSpace)
=
1-d_{\mathrm{TV}}(P,Q),
\]
which is the largest possible agreement probability over all couplings of
\(P\) and \(Q\). We note that this construction requires sampling not only
from the common component \(\gamma\), but also from the residual measures
\(P-\gamma\) and \(Q-\gamma\). In practice, when \(P\) and \(Q\) admit densities \(p\) and \(q\) with respect
to the same base measure, the following rejection-based implementation
provides a general maximal coupling algorithm.

\begin{algorithm}[ht]
\SetAlgoLined
\DontPrintSemicolon
\SetKwIF{If}{ElseIf}{Else}{if}{:}{else if}{else}{end}
\SetKwFunction{Sample}{Sample}

\textbf{Input:}\;
Probability measures \(P\) and \(Q\) with densities \(p\) and \(q\) w.r.t. a common dominating measure\;
\textbf{Output:}\;
A coupled pair \((X^P,X^Q)\) with marginals \(P\) and \(Q\)\;

\tcp{Step 1: Sample from \(P\).}
Sample \(X\sim P\) and \(W\sim \mathrm{Unif}(0,1)\)\;

\If{\(W \le \min\left\{1,\dfrac{q(X)}{p(X)}\right\}\)}{
    \(X^P \gets X\)\;
    \(X^Q \gets X\)\;
    \Return{\((X^P,X^Q)\)}\;
}
\Else{
    \(X^P \gets X\)\;

    \tcp{Step 2: Sample from the residual part of \(Q\).}
    \Repeat{\(W^* > \min\left\{1,\dfrac{p(Y)}{q(Y)}\right\}\)}{
        Sample \(Y\sim Q\) and \(W^*\sim \mathrm{Unif}(0,1)\)\;
    }

    \(X^Q \gets Y\)\;
    \Return{\((X^P,X^Q)\)}\;
}

\caption{Maximal coupling of \(P\) and \(Q\)}
\label{alg:max_coupling}
\end{algorithm}
\subsection{Metropolis Hasting Algorithms}\label{appendix:mh}
We provide the Metropolis-Hasting algorithm for target distribution $\pi(\cdot)$ using the proposal kernel $K(\cdot|x)$ in Algorithm \ref{alg:mh_kernel}.
\begin{algorithm}[ht]
\SetAlgoLined
\DontPrintSemicolon
\SetKwIF{If}{ElseIf}{Else}{if}{:}{else if}{else}{end}
\SetKwFunction{Sample}{Sample}

\textbf{Input:} current state $x\in\XSpace$, target density $\pi$, proposal kernel $K(\mathrm dy | x)$ with density $k(y\mid x)$\;
\textbf{Output:} next state $x'\in\XSpace$\;

\tcp{Step 1: Propose a candidate move.}
Sample $y \sim K(\cdot| x)$\;

\tcp{Step 2: Compute the MH acceptance probability.}
\[
\alpha(x,y)
\gets
\min \left(1,\frac{\pi(y)k(x\mid y)}{\pi(x)k(y\mid x)}\right) .
\]
Sample $u\sim \mathrm{Unif}(0,1)$\;

\tcp{Step 3: Accept or reject the proposal.}
\If{$u\le \alpha(x,y)$}{
    $x'\gets y$\;
}
\Else{
    $x'\gets x$\;
}
\Return{$x'$}\;

\caption{One step of the Metropolis--Hastings algorithm}
\label{alg:mh_kernel}
\end{algorithm}

\subsection{Poisson Monte Carlo Algorithm} \label{app:solve_PP_opt}
Algorithm~\ref{alg:pfr_mu_proposal} gives the procedure used to solve the optimization problem over the countably infinite set of Poisson points in~\eqref{PP_opt}, following the version of \citet{theis2022algorithms}. 
\begin{algorithm}[ht]
\SetAlgoLined
\DontPrintSemicolon
\SetKwInOut{Input}{Input}
\SetKwInOut{Output}{Output}

\Input{Target measure \(P\), proposal measure \(\mu\), and constant
\(w_{\min}>0\) such that
\[
w_{\min}\le \inf_{x\in\XSpace} \frac{\mathrm d\mu}{\mathrm dP}(x).
\]}
\Output{Sample \(X\sim P\)}

Generate the atoms
\[
\Omega=\{(X_j,S_j)\}_{j\ge 1}
\]
of a marked Poisson process on \(\XSpace\times \mathbb R_+\) with intensity
\(\mu(\mathrm dx)\otimes \mathrm ds\)

Set \(j\gets 1\), \(s^\star\gets \infty\), and \(j^\star\gets 1\)\;

\While{\(\mathrm{true}\)}{
    \[
        \widetilde S_j
        \gets
        S_j
        \frac{\mathrm d\mu}{\mathrm dP}(X_j)
    \]

    \If{\(\widetilde S_j < s^\star\)}{
        \(s^\star \gets \widetilde S_j\)\;
        \(j^\star \gets j\)\;
    }

    \If{\(s^\star \le S_j w_{\min}\)}{
        \Return{\(X_{j^\star}\)}\;
    }

    \(j\gets j+1\)\;
}

\caption{PFR sampling with proposal measure \(\mu\), following \citet{theis2022algorithms}}
\label{alg:pfr_mu_proposal}
\end{algorithm}

\section{List-Level Theorems}\label{appendix:list}
\paragraph{Notation Notice.}
Unlike Poisson matching, the list-level coupling does not select atoms by
Poisson indices. For this reason, we switch to use subscripts to denote distribution
indices for conveninent.
\subsection{List-Level Coupling Representation}
We first relate \(G\) to our list-level objective, fix an ordering of the variables and define, for \(k=1,\dots,C-1\),
\[
A_k := \left\{X_k \in \{X_{k+1},\dots,X_C\}\right\}.
\]
The event \(A_k\) means that \(X_k\) does not introduce a new cluster relative to the suffix \((X_{k+1},\dots,X_C)\). Accordingly, if we define
\[
G_k \triangleq | \{ \text{unique values among } X_k, \dots, X_C \}|,
\]
then
\[
G_k = G_{k+1} + \mathbf{1}\!\left\{X_k \notin \{X_{k+1},\dots,X_C\}\right\}
= G_{k+1} + 1 - \mathbf{1}(A_k).
\]
Since \(G_C=1\), summing over \(k\) yields the pathwise identity
\[
G = C - \sum_{k=1}^{C-1} \mathbf{1}(A_k),
\]
and therefore
\[
\mathbb{E}[G]
=
C - \sum_{k=1}^{C-1}
\Pr\!\left(X_k \in \{X_{k+1},\dots,X_C\}\right).
\]

\subsection{Proof of Theorem \ref{theorem:list_coupling}} \label{proof_list}
We restate the theorem for convenience below. 

\begin{theoremrestatement}[List-Level Coupling]
Let $\mu,\nu_1,\dots,\nu_m$ be probability measures on a common measurable space, all absolutely continuous with respect to a common dominating measure $\lambda$. With the abuse of notation, we denote their densities with the same letter when no confusion arises.
\[
\mu(x)=\frac{\dd\mu}{\dd\lambda}(x),
\qquad
\nu_i(x)=\frac{\dd\nu_i}{\dd\lambda}(x).
\]
 Define the barycenter
\[
\bar{\nu}(x) \coloneqq \frac{1}{m}\sum_{i=1}^m \nu_i(x).\]
Then there exists a coupling such that
\[
\Pr\bigl(X\notin\{Y_1,\dots,Y_m\}\bigr)
= E_m(\mu\|\bar\nu),
\]
which is optimal. Moreover, this coupling can be implemented without computing $E_m(\mu\|\bar\nu)$.
\end{theoremrestatement}
\begin{proof}
Set
\[
\ell(x)\coloneqq \min\{\mu(x),m\bar\nu(x)\}.
\]
Then
\[
1-E_m(\mu\|\bar\nu)
=
\int \ell(x)\,d\lambda(x),
\]
since $\ell(x)=\mu(x)-(\mu(x)-m\bar\nu(x))_+.$
For each $j\in[m]$, define
\[
h_j(x)=
\begin{cases}
\nu_j(x), & \mu(x)>m\bar\nu(x),\\[1ex]
\mu(x)\dfrac{\nu_j(x)}{m\bar\nu(x)}, & \mu(x)\le m\bar\nu(x),\ m\bar\nu(x)>0,\\[1ex]
0, & \bar\nu(x)=0.
\end{cases}
\]
Then, we have $\sum_{j=1}^m h_j(x)=\ell(x)$ for every $x$. Indeed, if $\mu(x)>m\bar\nu(x)$, then
\[
\sum_{j=1}^m h_j(x)=\sum_{j=1}^m \nu_j(x)=m\bar\nu(x)=\ell(x),
\]
whereas if $\mu(x)\le m\bar\nu(x)$ and $m\bar\nu(x)>0$, then
\[
\sum_{j=1}^m h_j(x)
=
\mu(x)\frac{\sum_{j=1}^m \nu_j(x)}{m\bar\nu(x)}
=
\mu(x)
=
\ell(x).
\]
When $\bar\nu(x)=0$, both sides are $0$.
Let
\[
\alpha_j\coloneqq \int h_j(x)\,d\lambda(x).
\]
Since $0\le h_j\le \nu_j$ and $\int \nu_j\,d\lambda=1$, we have $0\le \alpha_j\le 1$. Define the residual density
\[
r_j(x)\coloneqq
\begin{cases}
\dfrac{\nu_j(x)-h_j(x)}{1-\alpha_j}, & \alpha_j<1,\\[2ex]
0, & \alpha_j=1.
\end{cases}
\]
Let $R_j$ denote the corresponding probability measure when $\alpha_j<1$; if $\alpha_j=1$, then $\nu_j=h_j$ a.e. and the residual part is never needed.
Observe that
\[
\nu_j(x)-h_j(x)
=
\nu_j(x)\Bigl(1-\frac{\mu(x)}{m\bar\nu(x)}\Bigr)_+,
\]
with the convention that the right-hand side is $0$ when $\bar\nu(x)=0$. Indeed, if $\mu(x)>m\bar\nu(x)$ then $h_j(x)=\nu_j(x)$, so both sides are $0$; if $\mu(x)\le m\bar\nu(x)$ and $m\bar\nu(x)>0$, then
\[
\nu_j(x)-h_j(x)
=
\nu_j(x)-\mu(x)\frac{\nu_j(x)}{m\bar\nu(x)}
=
\nu_j(x)\Bigl(1-\frac{\mu(x)}{m\bar\nu(x)}\Bigr).
\]
Now construct the coupling as follows.
Sample $X\sim \mu$ and independently $U\sim\mathrm{Unif}(0,1)$. Define the hit event
\[
H\coloneqq \left\{U\le \min\Bigl(1,\frac{m\bar\nu(X)}{\mu(X)}\Bigr)\right\}.
\]
On $H$, choose an index $I\in\{1,\dots,m\}$ conditionally on $X$ according to
\[
\Pr(I=i\mid X)=\frac{\nu_i(X)}{m\bar\nu(X)}.
\]
This is well-defined on $H$, because $H$ implies $m\bar\nu(X)>0$ whenever $\mu(X)>0$. Then set $Y_I=X$. For the remaining coordinates $j\neq I$, generate random variables with marginals $R_j$. On $H^c$, generate all coordinates $Y_j$ with marginals $R_j$.
Notice that $Y_j$ drawn from the residual laws need not be independent; any joint coupling with the specified marginals is allowed. Independence is merely a convenient implementation choice.

We now verify the properties of this construction.

\paragraph{Membership probability.}
Since $X\sim \mu$, we can write
\[
\Pr(H)
=
\int \mu(x)\min\Bigl(1,\frac{m\bar\nu(x)}{\mu(x)}\Bigr)\,\dd\lambda(x)
=
\int \ell(x)\,\dd\lambda(x)
=
1-E_m(\mu\|\bar\nu).
\]
On the event $H$, one coordinate is set equal to $X$, namely $Y_I=X$, hence
\( \Pr\bigl(X\in\{Y_1,\dots,Y_m\}\bigr)\ge \Pr(H).\)
Conversely, on $H^c$ we have
\[
\mu(X)>m\bar\nu(X).
\]
For any $j$, if $\mu(x)>m\bar\nu(x)$ then $h_j(x)=\nu_j(x)$, so $\nu_j(x)-h_j(x)=0,$
and hence $r_j(x)=0$ whenever $\mu(x)>m\bar\nu(x)$. 
Therefore each residual law $R_j$ is supported on $\{x:\mu(x)\le m\bar\nu(x)\}$,
so on $H^c$ none of the residual draws can equal $X$. Thus
\[
\Pr\bigl(X\in\{Y_1,\dots,Y_m\}\bigr)=\Pr(H)=1-E_m(\mu\|\bar\nu),
\]

\paragraph{Correct marginals.}
Fix $i\in[m]$. We need to show that $Y_i\sim \nu_i$. To do that, note that for measurable $A$,
\[
\Pr(X\in A,\ H,\ I=i)
=
\int_A
\mu(x)\min\Bigl(1,\frac{m\bar\nu(x)}{\mu(x)}\Bigr)
\frac{\nu_i(x)}{m\bar\nu(x)}
\,\dd\lambda(x).
\]
The integrand simplifies to $h_i(x)$, because:

- if $\mu(x)\le m\bar\nu(x)$ and $m\bar\nu(x)>0$, then
\[
\mu(x)  \frac{\nu_i(x)}{m\bar\nu(x)}
=
\mu(x)\frac{\nu_i(x)}{m\bar\nu(x)}
=
h_i(x);
\]

- if $\mu(x)>m\bar\nu(x)$, then
\[
\mu(x) \frac{m\bar\nu(x)}{\mu(x)}\frac{\nu_i(x)}{m\bar\nu(x)}
=
\nu_i(x)
=
h_i(x).
\]

Hence
\[
\Pr(X\in A,\ H,\ I=i)=\int_A h_i(x)\,\dd\lambda(x).
\]
Since $Y_i=X$ on the event $\{H,I=i\}$, this event contributes exactly the measure with density $h_i$ to the law of $Y_i$.

On the complementary event $\{H,I\neq i\}\cup H^c,$ the coordinate $Y_i$ is drawn from the residual law $R_i$. The probability of this complementary event is
\[
1-\Pr(H,I=i)=1-\alpha_i,
\]
because $\Pr(H,I=i)=\int h_i(x)\,d\lambda(x)=\alpha_i.$ Therefore the contribution of this complementary event to the law of $Y_i$ is
\[
(1-\alpha_i)R_i,
\]
whose density is $(1-\alpha_i)r_i(x)=\nu_i(x)-h_i(x).$

Combining the two contributions, the law of $Y_i$ has density $h_i(x)+\nu_i(x)-h_i(x)=\nu_i(x).$ Thus $Y_i\sim \nu_i$ for every $i$.

\medskip

\paragraph{Optimality.}
We shall show that $\E_m(\mu\|bar\nu)$ is in fact the optimal value for the list membership probability. To that goal, define for any measurable $A$
\[
\Gamma(A)\coloneqq \Pr\bigl(X\in A,\ X\in\{Y_1,\dots,Y_m\}\bigr).
\]
Then $\Gamma$ is a finite measure and clearly have  $\Gamma(A)\le \Pr(X\in A)=\mu(A).$
Also, we can write 
\begin{align*}
\Gamma(A)
&=
\Pr\Bigl(X\in A,\ \bigcup_{i=1}^m\{X=Y_i\}\Bigr) \\
&\le
\sum_{i=1}^m \Pr(X\in A,\ X=Y_i) \\
&\le
\sum_{i=1}^m \Pr(Y_i\in A) \\
&=
\sum_{i=1}^m \nu_i(A)
=
m\bar\nu(A).
\end{align*}
Hence, $\Gamma$ is dominated by both $\mu$ and $m\bar \nu$, so its density satisfies $\frac{\dd\Gamma}{\dd\lambda}(x) \leq \min\{\mu(x), m\bar\nu(x)\}$ a.e., implying that 
\[
\Pr\bigl(X\in\{Y_1,\dots,Y_m\}\bigr)
=
\Gamma(\mathcal X)
=
\int \frac{\dd\Gamma}{\dd\lambda}(x)\,\dd\lambda
\le
\int \min\{\mu(x),m\bar\nu(x)\}\,\dd\lambda(x).
\]
Equivalently,
\[
\Pr\bigl(X\notin\{Y_1,\dots,Y_m\}\bigr)
\ge
1-\int \min\{\mu(x),m\bar\nu(x)\}\,\dd\lambda(x)
=
E_m(\mu\|\bar\nu).
\]
\end{proof}

\section{Optimality for Shifted-Exponential Setting}\label{example:pml_shifted_exp}
Let \(P^0=\mathrm{Exp}(1)\) and, for \(i=1,\ldots,m\), let
\(P^i=i + \mathrm{Exp}(1)\), where all exponentials have scale \(1\). Thus
\[
p_0(x)=e^{-x}\mathbf 1\{x\ge 0\},
\qquad
p_i(x)=e^{-(x-i)}\mathbf 1\{x\ge i\},
\qquad i=1,\ldots,m .
\]

\paragraph{Poisson matching is optimal.}
Let \(\mu\) be any common dominating measure for \(P^0,\ldots,P^m\), and write
\[
f_i(x)=\frac{dP^i}{d\mu}(x),
\qquad i=0,\ldots,m .
\]
Let
\[
\Pi=\{(T_n,W_n)\}_{n\ge 1}
\]
be the shared Poisson process on \(\mathbb R_+\times \mathbb R\) with intensity
\(dt\,\mu(dw)\). Under the Poisson matching rule, the output corresponding to
\(P^i\) is
\[
X^i=W_{I_i},
\qquad
I_i=\arg\min_{n\ge 1}\frac{T_n}{f_i(W_n)} .
\]

The key observation is that, on \([1,\infty)\),
\[
P^1(dx)=e\,P^0(dx),
\]
and therefore
\[
f_1(x)=e f_0(x),
\qquad x\ge 1,
\]
\(\mu\)-almost everywhere. Also, \(f_1(x)=0\) on \((-\infty,1)\).

Let \((T_*,X^0)\) be the atom selected by \(P^0\), and define
\[
R_*=\frac{T_*}{f_0(X^0)} .
\]
By definition of \(I_0\),
\begin{equation}
\frac{T_n}{f_0(W_n)}\ge R_*
\qquad
\text{for all } n\ge 1 .
\label{eq:proof_point_chosen}
\end{equation}

If \(X^0<1\), then \(X^0\) is outside the support of every \(P^i\),
\(i\ge 1\). Hence no \(X^i\), \(i\ge 1\), can equal \(X^0\), and therefore
\[
X^0\notin\{X^1,\ldots,X^m\}.
\]

If \(X^0\ge 1\), then \(f_1(X^0)=e f_0(X^0)\), so the \(P^1\)-score of the same
atom is
\[
\frac{T_*}{f_1(X^0)}
=
\frac{T_*}{e f_0(X^0)}
=
\frac{R_*}{e}.
\]
For any other atom with \(W_n\ge 1\), we have
\[
\frac{T_n}{f_1(W_n)}
=
\frac{T_n}{e f_0(W_n)}
\ge
\frac{R_*}{e},
\]
where the last inequality follows from~\eqref{eq:proof_point_chosen}. Atoms
with \(W_n<1\) have \(f_1(W_n)=0\), so they cannot be selected by \(P^1\).
Therefore \(P^1\) selects the same atom as \(P^0\), and hence
\[
X^1=X^0.
\]

Consequently,
\[
X^0\in\{X^1,\ldots,X^m\}
\quad\Longleftrightarrow\quad
X^0\ge 1.
\]
Since \(X^0\sim P^0=\mathrm{Exp}(0)\),
\[
\Pr(X^0\ge 1)=e^{-1}.
\]
Thus, under the shared Poisson matching construction,
\[
\Pr\!\left(X^0\notin\{X^1,\ldots,X^m\}\right)
=
\Pr(X^0<1)
=
1-e^{-1}.
\]
or 
\[
\Pr\!\left(X^0\in\{X^1,\ldots,X^m\}\right)
=
e^{-1}.
\]

\paragraph{Hockey-stick lower bound.}
It turns out that the quantity above is also the desired optimal Hockey-Stick divergence.
\[
\bar P \coloneqq \frac1m\sum_{i=1}^m P^i
\]
be the barycenter of \(P^1,\ldots,P^m\). The order-\(m\) hockey-stick
divergence between \(P^0\) and \(\bar P\) is
\[
E_m(P^0\|\bar P)
=
\int \left(p_0(x)-m\bar p(x)\right)_+\,dx
=
\int \left(p_0(x)-\sum_{i=1}^m p_i(x)\right)_+\,dx .
\]
Using
\[
p_0(x)=e^{-x}\mathbf 1\{x\ge 0\},
\qquad
p_i(x)=e^{-(x-i)}\mathbf 1\{x\ge i\},
\]
we have
\[
p_0(x)-\sum_{i=1}^m p_i(x)
=
e^{-x}\mathbf 1\{x\ge 0\}
-
\sum_{i=1}^m e^{-(x-i)}\mathbf 1\{x\ge i\}.
\]
For \(0\le x<1\), none of \(p_1,\ldots,p_m\) are active, so
\[
p_0(x)-\sum_{i=1}^m p_i(x)=e^{-x}>0.
\]
For \(x\ge 1\), the first shifted exponential is active and satisfies
\[
p_1(x)=e^{-(x-1)}=e\,p_0(x)>p_0(x).
\]
Hence
\[
p_0(x)-\sum_{i=1}^m p_i(x)<0,
\qquad x\ge 1.
\]
Therefore the positive part is supported only on \([0,1)\), and
\[
E_m(P^0\|\bar P)
=
\int_0^1 e^{-x}\,dx
=
1-e^{-1}.
\]
Thus the optimal list-level inclusion probability is
\[
1-E_m(P^0\|\bar P)
=
e^{-1}.
\]
This concludes that, Poisson matching achieves the optimal inclusion probability. 
\paragraph{Optimal value of \(\mathbf G^*\).}
Fix the left-to-right ordering \(P^0,P^1,\ldots,P^m\), and let \(C=m+1\).
Recall that
\[
\mathbf G^*
=
C-
\sup_{\gamma\in\Gamma}
\sum_{k=0}^{m-1}
\Pr_\gamma\!\left(
X^k\in\{X^{k+1},\ldots,X^m\}
\right),
\]
where \(\Gamma\) denotes the set of all couplings of
\(P^0,\ldots,P^m\).

For the shifted-exponential family, the same argument as above applies to every
\(k=0,\ldots,m-1\). Indeed, \(P^{k+1}\) is just a unit shift of \(P^k\), and on
the overlap region \([k+1,\infty)\) its density is \(e\) times that of \(P^k\).
Therefore, under the shared Poisson matching construction,
\[
\Pr\!\left(
X^k\in\{X^{k+1},\ldots,X^m\}
\right)
=
e^{-1}.
\]
The corresponding order-\((m-k)\) hockey-stick bound has the same value, so
Poisson matching attains the optimal list-level probability for every term in
the left-to-right decomposition. Hence
\[
\sup_{\gamma\in\Gamma}
\sum_{k=0}^{m-1}
\Pr_\gamma\!\left(
X^k\in\{X^{k+1},\ldots,X^m\}
\right)
=
m e^{-1}.
\]
Since \(C=m+1\), we obtain
\[
\mathbf G^*
=
C-(C-1)e^{-1}.
\]
Thus the lower bound is also achievable, using the Poisson Matching construction.
\section{Proof for the Upperbound of Theorem \ref{theorem:general_bounds}}\label{appendix:proof_bnd_G}


The upperbound of interest can be restated as follow. 

\begin{theoremrestatement}
Let \(X^1,\dots,X^C\) be random variables defined previously, and define
\[
G:=\bigl|\{\text{unique values among } X^1,\dots,X^C\}\bigr|,
\qquad
D:=\sum_{1\le i<j\le C}\mathbf 1\{X^i\neq X^j\}.
\]
Then
\begin{align}
    \mathbb E[G]
\le
\frac{1+\sqrt{1+8\,\mathbb E[D]}}{2}
=
\frac{1+\sqrt{1+8\sum_{1\le i<j\le C}\Pr (X^i\neq X^j)}}{2}. 
\end{align}

Furthermore, there exists a coupling scheme achieves:
\begin{align}
 \mathbb E[G]
\le
\frac{
1+\sqrt{
1+8\sum_{1\le i<j\le C}
\frac{2\,d_{\mathrm{TV}}(P^i,P^j)}{1+d_{\mathrm{TV}}(P^i,P^j)} 
}
}{2}.   
\end{align}
\end{theoremrestatement}\
\begin{proof}
For any realization of \(X^1,\dots,X^C\), let \(G\) denote the number of unique realized values among \(X^1,\dots,X^C\). Then the realization induces \(G\) distinct clusters. Any two distinct clusters contribute at least one mismatching pair, and since there are \(\binom{G}{2}\) unordered pairs of distinct clusters, we have
\[
D \ge \binom{G}{2}.
\]
Equivalently,
\[
\frac{G(G-1)}{2}\le D,
\]
which implies
\[
G^2-G-2D\le 0.
\]
Solving this quadratic inequality yields
\[
G \le \frac{1+\sqrt{1+8D}}{2}.
\]
Taking expectations and using the concavity of the square-root function gives
\[
\mathbb E[G]
\le
\mathbb E\!\left[\frac{1+\sqrt{1+8D}}{2}\right]
\le
\frac{1+\sqrt{1+8\,\mathbb E[D]}}{2}.
\]
Finally, by linearity of expectation,
\[
\mathbb E[D]
=
\sum_{1\le i<j\le C}\mathbb \Pr(X^i\neq X^j),
\]
which proves the claim. Applying the bound in (\ref{pml_bound_tv}) gives the second bound, since all the samplers use the same shared marked Poisson point process $\Omega$. Finally, the optimal $\mathbf{G}^* \leq \mathbb{E}[G]$ thus the upperbound in Theorem \ref{theorem:general_bounds} holds.
\end{proof}

\section{Proof of Lemma \ref{lemma:barycenter_PML}}\label{proof:barycenter_PML}

\begin{lemmarestatement}[Restatement of Lemma~\ref{lemma:barycenter_PML}]
Let \(P^1,\dots,P^C\) be probability measures on \(\XSpace\), and the marked Poisson process defined previously with the common barycenter measure $\mu$ above. Then:
\begin{equation}
    P^i \ll \mu
    \qquad\text{and}\qquad
    \operatorname*{ess\,sup}_{x \in \XSpace }
    \left(\frac{\mathrm dP^i}{\mathrm d\mu}\right)(x)
    \le C,
    \qquad i=1,\dots,C .
\end{equation}
As a result, sampling from each measure has worst-case complexity \(O(C)\), and since they all share the same marked Poisson process, the expected number of samples per probability measure is \(O(1)\).
\end{lemmarestatement}
\begin{proof}
    Fix $i \in \{1,\dots,C\}$. By construction,
\begin{align}
    \mu(A)
    = \frac{1}{C}\sum_{j=1}^C P^j(A)
    \ge \frac{1}{C} P^i(A)
\end{align}
for every measurable set $A \subseteq \XSpace$. Therefore, if $\mu(A)=0$, then necessarily $P^i(A)=0$, which shows that $P^i \ll \mu$.

Since $P^i$ is absolutely continuous with respect to $\mu$, the Radon--Nikodym derivative $\mathrm dP^i/\mathrm d\mu$ is well-defined. Moreover, the measure inequality
\begin{align}
    P^i(A) \le C\,\mu(A)
\end{align}
holds for every measurable set $A$. This implies
\begin{align}
    \frac{\mathrm dP^i}{\mathrm d\mu}(x) \le C
\end{align}
for $\mu$-almost every $x$. Hence,
\begin{align}
    \operatorname*{ess\,sup}_{x \in \XSpace } \frac{\mathrm d P^i}{\mathrm  d\mu}(x) \le C,
\end{align}
which is also the expected runtime complexity of PMC, concluding the proof.
\end{proof}
\section{Two-Step Metropolis Hasting Coupling Algorithms}\label{appendix:coupling_mh}

We provide the 2-step MH coupling algorithms for comparision in Algorithm \ref{alg:two_stage_mh_coupling}.

\begin{algorithm}[H]
\SetAlgoLined
\DontPrintSemicolon
\SetKwInOut{Input}{Input}
\SetKwInOut{Output}{Output}

\Input{Current states $x_t^{(1)},\dots,x_t^{(C)}$}
\Output{Updated states $x_{t+1}^{(1)},\dots,x_{t+1}^{(C)}$}

Set $K_i(\mathrm dy)\triangleq k(y\mid x_t^{(i)})\,\mathrm dy$, $i=1,\dots,C$, and
\[
\mu(\mathrm dy)
=
\frac{1}{C}\sum_{i=1}^C K_i(\mathrm dy).
\]

\tcp{Step 1: PML proposal coupling}
Generate $\Omega=\{(S_j, Y_j,)\}_{j\ge 1}$ with intensity measure $\mathrm ds\otimes\mu$\;

\For{$i=1,\dots,C$}{
    \[
    J^{(i)}
    =
    \arg\min_{j\ge 1}
    \left\{
    S_j
    \left(
    \frac{\mathrm d K_i}{\mathrm d\mu}(Y_j)
    \right)^{-1}
    \right\}
    \]
    \[
    Y_t^{(i)} \gets Y_{J^{(i)}}
    \]
}

\tcp{Step 2: Shared accept--reject}
$U_t \sim \mathrm{Unif}(0,1)$\;

\For{$i=1,\dots,C$}{
    \eIf{$U_t \le \alpha(x_t^{(i)},Y_t^{(i)})$}{
        $x_{t+1}^{(i)} \gets Y_t^{(i)}$\;
    }{
        $x_{t+1}^{(i)} \gets x_t^{(i)}$\;
    }
}

\caption{Two-stage coupled kernel for Metropolis--Hastings chains}
\label{alg:two_stage_mh_coupling}
\end{algorithm}
\section{Two-Stage Coupled Kernel}\label{app:two_stage}
This construction is given in Algorithm~\ref{alg:two_stage_mh_coupling}, Appendix ~\ref{appendix:coupling_mh}, where in the first step, we couple the proposal kernels \(K(\dd y\cdot \mid x_t^{(i)}) \triangleq k(y|x) \dd y\) using the adaptive Poisson matching method developed in Section \ref{sec:accel_pml}. In the second step, we couple the accept/reject decisions using a shared uniform random variable which achieves optimal pairwise matching probability for Bernouli random variables.

\paragraph{Pairwise Matching Probability.}
For any two distinct chains \(x_t^{(a)} \neq x_t^{(b)}\), assume that the
proposal kernel admits a density \(k(x,y)\) with respect to Lebesgue measure.
Then the event
\[
\{x_{t+1}^{(a)} = x_{t+1}^{(b)}\}
\]
can occur with nonzero probability only if both chains obtain and accept the
same proposed value. Indeed, the alternatives
\[
(x_{t+1}^{(a)},x_{t+1}^{(b)})
=
(x_t^{(a)},x_t^{(b)}),
\qquad
(x_{t+1}^{(a)},x_{t+1}^{(b)})
=
(Y_t^{(a)},x_t^{(b)}),
\qquad
(x_{t+1}^{(a)},x_{t+1}^{(b)})
=
(x_t^{(a)},Y_t^{(b)})
\]
cannot yield equality except on null events, such as
\(Y_t^{(a)}=x_t^{(b)}\) or \(Y_t^{(b)}=x_t^{(a)}\), which have probability
zero by continuity. Hence, writing
\[
\mathbb{P}_{\mathrm{2\text{-}step}}
\coloneqq
\Pr\!\left(x_{t+1}^{(a)} = x_{t+1}^{(b)}\right),
\]
we have
\[
\mathbb{P}_{\mathrm{2\text{-}step}}
\geq
\mathbb{P}_{\mathrm{2,bound}},
\]
where
\begin{align}
\mathbb{P}_{\mathrm{2,bound}}
&\coloneqq
\int
\min\!\left\{
\alpha(x_t^{(a)},y),
\alpha(x_t^{(b)},y)
\right\}
\frac{
k(y |x_t^{(a)})k(y | x_t^{(b)})
}{
k(y | x_t^{(a)})+k(y | x_t^{(b)})
}
\,\mathrm dy .
\end{align}
See Appendix \ref{appendix:proof_compare_bound} for comparison to the joint method. To show the above bound, consider the following steps:
\begin{align}
\mathbb{P}_{\mathrm{2\text{-}step}}
&= \Pr\!\left(x_{t+1}^{(a)} = x_{t+1}^{(b)}\right) \\
&= \Pr\!\left(
U_t \le \alpha(x_t^{(a)},Y_t^{(a)}),\;
U_t \le \alpha(x_t^{(b)},Y_t^{(b)}),\;
Y_t^{(a)} = Y_t^{(b)}
\right) \\
&= \int_{y \notin \{x_t^{(a)},x_t^{(b)}\}}
\min\!\left\{\alpha(x_t^{(a)},y),\alpha(x_t^{(b)},y)\right\}
\, \Pr\!\left(Y_t^{(a)}=Y_t^{(b)}\in \mathrm dy\right) \\
&\ge
\int_{y \notin \{x_t^{(a)},x_t^{(b)}\}}
\min\!\left\{\alpha(x_t^{(a)},y),\alpha(x_t^{(b)},y)\right\}
\frac{k(y|x_t^{(a)})k(y | x_t^{(b)})}
{k(y|x_t^{(a)})+k(y | x_t^{(b)})}
\,\mathrm dy \\
&\triangleq \mathbb{P}_{\mathrm{2,bound}}.
\end{align}
Here, $\Pr\!\left(Y_t^{(a)}=Y_t^{(b)}\in \mathrm dy\right)$ denotes the diagonal measure induced by the PML coupling. The inequality is due to Poisson Matching Lemma.

\section{Matching Probability of Joint Coupled Kernel}\label{appendix:mh_matching_prob_1step}

\paragraph{Pairwise Matching Probability.}
For two chains indexed by \(a\) and \(b\), the resulting matching probability
\(\mathbb{P}_{\mathrm{direct}} \triangleq \Pr\!\left(x_{t+1}^{(a)} = x_{t+1}^{(b)}\right)\) admits the following lower bound \(\mathbb{P}_{\mathrm{1,bound}}\):
\begin{align}
\mathbb{P}_{\mathrm{1,bound}} =
\int
\frac{
\alpha(x_t^{(a)},y)k(y\mid x_t^{(a)})\,
\alpha(x_t^{(b)},y)k(y\mid x_t^{(b)})
}{
\alpha(x_t^{(a)},y)k(y\mid x_t^{(a)})
+
\alpha(x_t^{(b)},y)k(y\mid x_t^{(b)})
}
\,\mathrm dy .
\end{align}

In particular,
\begin{align}
\mathbb{P}_{\mathrm{direct}}
&\triangleq \Pr\!\left(x_{t+1}^{(a)} = x_{t+1}^{(b)}\right) \\
&= \Pr\!\left(Y_t^{(a)} = Y_t^{(b)},\; U_t^{(a)} = U_t^{(b)} = 1\right) \\
&= \int \Pr\!\left(Y_t^{(a)} = Y_t^{(b)} = y,\; U_t^{(a)} = U_t^{(b)} = 1\right)\,\mathrm dy  \\ 
&= \int \Pr\!\left(Y_t^{(a)} = Y_t^{(b)},\; U_t^{(a)} = U_t^{(b)} \mid Y_t^{(a)} = y, U_t^{(a)} = 1\right)
\,\alpha(x_t^{(a)},y)k(y\mid x_t^{(a)})\,\mathrm dy .
\end{align}

Applying PML, we have the lower bound
\begin{align}
\mathbb{P}_{\mathrm{direct}}
&\ge
\int
\frac{
\alpha(x_t^{(a)},y)k(y\mid x_t^{(a)})\,
\alpha(x_t^{(b)},y)k(y\mid x_t^{(b)})
}{
\alpha(x_t^{(a)},y)k(y\mid x_t^{(a)})
+
\alpha(x_t^{(b)},y)k(y\mid x_t^{(b)})
}
\,\mathrm dy \\
&\triangleq \mathbb{P}_{\mathrm{1,bound}}.
\end{align}
\subsection{Compare Matching Probability Bounds}\label{appendix:proof_compare_bound}
\begin{theorem}\label{thm:compare_bound}
Let \(\mathbb{P}_{\mathrm{1,bound}}\) denote the matching probability bound obtained using PMC applied directly to the Metropolis--Hastings transition kernel, and let \(\mathbb{P}_{\mathrm{2,bound}}\) denote the matching probability induced by the two-step construction described above. Then we have
$
\mathbb{P}_{\mathrm{1,bound}} \;\ge\; \mathbb{P}_{\mathrm{2,bound}} .
$
\end{theorem}
\begin{proof}
It is sufficient to prove that for every \(y\),
\[
\frac{
\alpha(x^A,y) k(y\mid x^A)\,
\alpha(x^B,y) k(y\mid x^B)
}{
\alpha(x^A,y) k(y\mid x^A)
+
\alpha(x^B,y) k(y\mid x^B)
}
\ge
\min\{\alpha(x^{A},y),\alpha(x^{B},y)\}
\frac{
k(y\mid x^A)k(y\mid x^B)
}{
k(y\mid x^A)+k(y\mid x^B)
}.
\]
Using the shorthand
\[
\alpha^A=\alpha(x^A,y),\qquad
\alpha^B=\alpha(x^B,y),\qquad
k^A=k(y\mid x^A),\qquad
k^B=k(y\mid x^B),
\]
and assuming without loss of generality that \(\alpha^A\le \alpha^B\), the
inequality becomes
\[
\frac{\alpha^A k^A\,\alpha^B k^B}
{\alpha^A k^A+\alpha^B k^B}
\ge
\alpha^A\,\frac{k^A k^B}{k^A+k^B}.
\]
If \(\alpha^A k^A k^B=0\), the inequality is trivial. Otherwise, canceling the
positive factor \(\alpha^A k^A k^B\) yields
\[
\frac{\alpha^B}{\alpha^A k^A+\alpha^B k^B}
\ge
\frac{1}{k^A+k^B}.
\]
Since all terms are nonnegative, cross-multiplication gives
\[
\alpha^B(k^A+k^B)
\ge
\alpha^A k^A+\alpha^B k^B
\Longleftrightarrow
\alpha^B k^A
\ge
\alpha^A k^A,
\]
which holds because \(k^A\ge 0\) and \(\alpha^B\ge \alpha^A\). This completes
the proof.
\end{proof}
\section{Proof of Lemma \ref{lem:match_prob}} \label{aug_proof}

\begin{lemma}
Consider two chains at states $x^1_t$ and $x^2_t$ respectively. At $t+1$, let $\mathbb{P}_{\mathrm{aug}}(x^1_t,x^2_t)$ denote the meeting probability under the product-augmented construction, and let $\mathbb{P}_{\mathrm{direct}}(x^1_t,x^2_t)$ denote the meeting probability under the original direct construction. Then
\[
\mathbb{P}_{\mathrm{aug}}(x^1_t,x^2_t) \geq \mathbb{P}_{\mathrm{direct}}(x^1_t,x^2_t).
\]
\end{lemma}
\begin{proof}
We will show later that, for every \(u\) in the continuous accepted support, one can write the following conditional matching probability for the direct matching strategy:
\begin{align}
   &\Pr_{\mathrm{direct}}(X^2_{t+1} = x | X^1_{t+1} = x, X^1_t = x^1_t, X^2_t = x^2_t) = \frac{1}{1 +\beta(x)} 
\end{align}
For the augmented strategy:
\begin{align}
   &\Pr_{\mathrm{aug}}(X^2_{t+1} = x | X^1_{t+1} = x, X^1_t = x^1_t, X^2_t = x^2_t)\\ &=  \Pr_{\mathrm{aug}}(Y^2_{t+1} = x, U^2_{t+1}=1, | Y^1_{t+1} = x, U^1_{t+1}=1,  X^1_t = x^1_t, X^2_t = x^2_t)\\
   &=\frac{1}{1 +\beta'(1,x)}
\end{align}
where the exact form of $\beta, \beta'$ will be provided later. Furthermore, we have
\[
\beta'(1,x)\leq \beta(x).
\]
Therefore,
\begin{equation}
    \frac{1}{1+\beta'(1,x)}
    \geq
    \frac{1}{1+\beta(x)} .
    \label{ineq1}
\end{equation}
So we have, for any $x$:
\begin{equation}
    \Pr_{\mathrm{aug}}(X^2_{t+1} = x | X^1_{t+1} = x, X^1_t = x^1_t, X^2_t = x^2_t) \geq \Pr_{\mathrm{direct}}(X^2_{t+1} = x | X^1_{t+1} = x, X^1_t = x^1_t, X^2_t = x^2_t) \notag
\end{equation}
and since $P(X^1_{t+1} = x| X^1_t = x^1_t, X^2_t = x^2_t)$ is the same these two chains, integrate w.r.t. this gives the desired inequality.
\end{proof}

\subsection{Proof of conditional matching probabilities}
We prove the result by directly comparing the two matching probabilities. 
For notational simplicity, we use a different notation in this proof. With a slight abuse of notation, we preserve the subscripts \(a\) and \(b\) for the two chains while omitting the explicit dependence on time \(t\), i.e. let \(x^1_t = x_a\) and \(x^2_t = x_b\) denote their current states. 

Assume \(x_a\neq x_b\), and let \(\lambda\) be a
common dominating measure for the proposal parts of the two kernels. For the direct kernel between two chains at state $x_a$ and $x_b$, define
\begin{align}
    P(\mathrm dy)
    &\triangleq \kappa(\mathrm dy\mid x_a)
     = \alpha(x_a,y)K(\mathrm dy\mid x_a)
     + r(x_a)\delta_{x_a}(\mathrm dy), \text{ (chain $a$ kernel)}\notag\\
    Q(\mathrm dy)
    &\triangleq \kappa(\mathrm dy\mid x_b)
     = \alpha(x_b,y)K(\mathrm dy\mid x_b)
     + r(x_b)\delta_{x_b}(\mathrm dy), \text{ (chain $b$ kernel)}\notag\\
    P_a &\triangleq r(x_a),
    \qquad
    Q_b \triangleq r(x_b), \notag\\
    p(y)\lambda(\mathrm dy)
    &\triangleq \alpha(x_a,y)K(\mathrm dy\mid x_a),
    \qquad
    q(y)\lambda(\mathrm dy)
    \triangleq \alpha(x_b,y)K(\mathrm dy\mid x_b), \notag\\
    a(y)\lambda(\mathrm dy)
    &\triangleq \bigl(1-\alpha(x_a,y)\bigr)K(\mathrm dy\mid x_a),
    \qquad
    b(y)\lambda(\mathrm dy)
    \triangleq \bigl(1-\alpha(x_b,y)\bigr)K(\mathrm dy\mid x_b).
    \notag
\end{align}
Thus
\[
P(\mathrm dy)=P_a\delta_{x_a}(\mathrm dy)+p(y)\lambda(\mathrm dy),
\qquad
Q(\mathrm dy)=Q_b\delta_{x_b}(\mathrm dy)+q(y)\lambda(\mathrm dy).
\]

\paragraph{Direct Matching.} Under the notation above, the matching probability at kernel level is equivalent to that with $P$ and $Q$. In particular, recall the two probability measures
\[
P(\mathrm dy)=P_a\,\delta_{x_a}(\mathrm dy)+p(y)\lambda(\mathrm dy),
\qquad
Q(\mathrm dy)=Q_b\,\delta_{x_b}(\mathrm dy)+q(y)\lambda(\mathrm dy),
\]
where \(P_a=P(\{x_a\})\), \(Q_b=Q(\{x_b\})\), and \(p,q\) are the densities of
the non-atomic parts with respect to \(\lambda\). Hence
\[
\int p(v)\lambda(\mathrm dv)=1-P_a,
\qquad
\int q(v)\lambda(\mathrm dv)=1-Q_b.
\]
Let $X_P \sim P$ and $X_Q \sim Q$, then we have:
\begin{equation}
    \Pr(X_Q = u| X_P = u) = \Pr(X^2_{t+1}=u|X^1_{t+1}=u, X^2_{t}=x_b, X^1_{t}=x_a)
\end{equation}

For direct matching with the Poisson Matching Lemma (PML), we choose a proposal
measure \(\mu\) satisfying \(P,Q\ll \mu\). In particular, take
\[
\mu(\mathrm dy)
=
\eta_a\,\delta_{x_a}(\mathrm dy)
+
\eta_b\,\delta_{x_b}(\mathrm dy)
+
t(y)\lambda(\mathrm dy),
\]
where \(\eta_a,\eta_b>0\), and \(t(y)>0\) on the support of \(p+q\). Define
\[
f=\frac{\mathrm dP}{\mathrm d\mu},
\qquad
g=\frac{\mathrm dQ}{\mathrm d\mu}.
\]
Then
\[
f(x_a)=\frac{P_a}{\eta_a},
\qquad
f(x_b)=0,
\qquad
f(y)=\frac{p(y)}{t(y)},
\]
and
\[
g(x_a)=0,
\qquad
g(x_b)=\frac{Q_b}{\eta_b},
\qquad
g(y)=\frac{q(y)}{t(y)}.
\]
For \(u\) in the continuous support with \(p(u),q(u)>0\), define
\[
\beta(u)
\triangleq
f(u)\int
\left(
\frac{g(v)}{g(u)}
-
\frac{f(v)}{f(u)}
\right)_+
\mu(\mathrm dv).
\]
This is the quantity appearing in Eq.~(16) of \cite{li2021unified}, and the
conditional meeting probability is
\[
\Pr(\widetilde U_Q=\widetilde U_P\mid \widetilde U_P=u)
=
\frac{1}{1+\beta(u)},
\]
where \(\widetilde U_P\) and \(\widetilde U_Q\) are the values selected according
to the PML from the common marked point process
\(\Pi=\{(T_i,U_i)\}_{i=1}^{\infty}\), with \(U_i\sim\mu\). In our case, those are $X_P$ and $X_Q$ respectively.

We now compute \(\beta(u)\). Splitting the integral into the atoms \(x_a,x_b\)
and the non-atomic part gives
\[
\beta(u)
=
Q_b\frac{p(u)}{q(u)}
+
p(u)
\int
\left(
\frac{q(v)}{q(u)}
-
\frac{p(v)}{p(u)}
\right)_+
\lambda(\mathrm dv).
\]
Indeed, at \(v=x_a\), we have \(g(x_a)=0\), so the positive part vanishes. At
\(v=x_b\), we have \(f(x_b)=0\) and \(g(x_b)=Q_b/\eta_b\), hence the atomic
contribution is
\[
f(u)\,\eta_b\,\frac{g(x_b)}{g(u)}
=
\frac{p(u)}{t(u)}\,\eta_b\,
\frac{Q_b/\eta_b}{q(u)/t(u)}
=
Q_b\frac{p(u)}{q(u)}.
\]
For the non-atomic part,
\[
f(v)=\frac{p(v)}{t(v)},
\qquad
g(v)=\frac{q(v)}{t(v)}.
\]
Therefore
\begin{align*}
&f(u)
\int
\left(
\frac{g(v)}{g(u)}
-
\frac{f(v)}{f(u)}
\right)_+
t(v)\lambda(\mathrm dv) \\
&\qquad =
\frac{p(u)}{t(u)}
\int
\left[
\frac{t(u)}{t(v)}
\left(
\frac{q(v)}{q(u)}
-
\frac{p(v)}{p(u)}
\right)
\right]_+
t(v)\lambda(\mathrm dv) \\
&\qquad =
p(u)
\int
\left(
\frac{q(v)}{q(u)}
-
\frac{p(v)}{p(u)}
\right)_+
\lambda(\mathrm dv).
\end{align*}
Adding the atomic and non-atomic contributions yields the displayed formula for
\(\beta(u)\).
\paragraph{Augmented Matching.}
Now consider the augmented space \(\{0,1\}\times \mathcal{X}\), where the first
coordinate records whether the proposal is accepted. Define
\begin{align}
    P'(\mathrm dj,\mathrm dy)
    &\triangleq
    K(\mathrm dy\mid x_a)\,
    \mathrm{Ber}\bigl(\alpha(x_a,y)\bigr)(\mathrm dj), \notag\\
    Q'(\mathrm dj,\mathrm dy)
    &\triangleq
    K(\mathrm dy\mid x_b)\,
    \mathrm{Ber}\bigl(\alpha(x_b,y)\bigr)(\mathrm dj).
    \notag
\end{align}
Equivalently, using the notation introduced above,
\begin{align}
    P'(\{1\},\mathrm dy)
    &= p(y)\lambda(\mathrm dy),
    &
    Q'(\{1\},\mathrm dy)
    &= q(y)\lambda(\mathrm dy), \notag\\
    P'(\{0\},\mathrm dy)
    &= a(y)\lambda(\mathrm dy),
    &
    Q'(\{0\},\mathrm dy)
    &= b(y)\lambda(\mathrm dy).
    \notag
\end{align}
Compared to the kernel-level notation, our event of interest is:
\begin{align}
    \Pr(Y_Q = u, U_Q&=1| Y_P = u, U_P=1) \\ &= \Pr(Y^2_{t+1}=u, U^2_{t+1}=1|Y^1_{t+1}=u, U^1_{t+1}=1, X^2_{t}=x_b, X^1_{t}=x_a)
\end{align}
Let \(\mu'\) be a dominating reference measure such that \(P',Q'\ll \mu'\), and
define
\[
f'(j,y)=\frac{\mathrm dP'}{\mathrm d\mu'}(j,y),
\qquad
g'(j,y)=\frac{\mathrm dQ'}{\mathrm d\mu'}(j,y).
\]
We are interested in the accepted-proposal matching event, which corresponds to
the branch \(j=1\). Applying Eq.~(16) of \cite{li2021unified}, for
\(p(u),q(u)>0\), define
\[
\beta'(1,u)
\triangleq
f'(1,u)
\int
\left(
\frac{g'(j,v)}{g'(1,u)}
-
\frac{f'(j,v)}{f'(1,u)}
\right)_+
\mu'(\mathrm dj,\mathrm dv).
\]
The corresponding conditional meeting probability is
\[
\Pr(\widetilde U_{Q'}=\widetilde U_{P'}\mid \widetilde U_{P'}=(1,u))
=
\frac{1}{1+\beta'(1,u)}.
\]
Splitting the integral over the two branches \(j=0\) and \(j=1\), write
\[
\beta'(1,u)=I_0+I_1,
\]
where
\begin{align}
I_0
&\triangleq
f'(1,u)
\int
\left(
\frac{g'(0,v)}{g'(1,u)}
-
\frac{f'(0,v)}{f'(1,u)}
\right)_+
\mu'(\{0\},\mathrm dv), \notag\\
I_1
&\triangleq
f'(1,u)
\int
\left(
\frac{g'(1,v)}{g'(1,u)}
-
\frac{f'(1,v)}{f'(1,u)}
\right)_+
\mu'(\{1\},\mathrm dv).
\notag
\end{align}

First consider \(I_1\). By the definition of the Radon--Nikodym derivatives,
\[
f'(1,v)\mu'(\{1\},\mathrm dv)=p(v)\lambda(\mathrm dv),
\qquad
g'(1,v)\mu'(\{1\},\mathrm dv)=q(v)\lambda(\mathrm dv).
\]
Moreover, for \(p(u),q(u)>0\),
\[
\frac{f'(1,u)}{g'(1,u)}
=
\frac{p(u)}{q(u)}.
\]
Therefore
\begin{align}
I_1
&=
f'(1,u)
\int
\left(
\frac{g'(1,v)}{g'(1,u)}
-
\frac{f'(1,v)}{f'(1,u)}
\right)_+
\mu'(\{1\},\mathrm dv) \notag\\
&=
\int
\left(
\frac{f'(1,u)}{g'(1,u)}g'(1,v)
-
f'(1,v)
\right)_+
\mu'(\{1\},\mathrm dv) \notag\\
&=
\int
\left(
\frac{p(u)}{q(u)}q(v)-p(v)
\right)_+
\lambda(\mathrm dv) \notag\\
&=
p(u)
\int
\left(
\frac{q(v)}{q(u)}
-
\frac{p(v)}{p(u)}
\right)_+
\lambda(\mathrm dv).
\notag
\end{align}

Now consider \(I_0\). Similarly,
\[
f'(0,v)\mu'(\{0\},\mathrm dv)=a(v)\lambda(\mathrm dv),
\qquad
g'(0,v)\mu'(\{0\},\mathrm dv)=b(v)\lambda(\mathrm dv).
\]
Thus
\begin{align}
I_0
&=
f'(1,u)
\int
\left(
\frac{g'(0,v)}{g'(1,u)}
-
\frac{f'(0,v)}{f'(1,u)}
\right)_+
\mu'(\{0\},\mathrm dv) \notag\\
&=
\int
\left(
\frac{f'(1,u)}{g'(1,u)}g'(0,v)
-
f'(0,v)
\right)_+
\mu'(\{0\},\mathrm dv) \notag\\
&=
\int
\left(
\frac{p(u)}{q(u)}b(v)-a(v)
\right)_+
\lambda(\mathrm dv) \notag\\
&=
p(u)
\int
\left(
\frac{b(v)}{q(u)}
-
\frac{a(v)}{p(u)}
\right)_+
\lambda(\mathrm dv).
\notag
\end{align}
Using \((x-y)_+\leq x\) for \(x,y\geq 0\), we obtain
\begin{align}
I_0
&\leq
p(u)
\int
\frac{b(v)}{q(u)}
\lambda(\mathrm dv) \notag\\
&=
\frac{p(u)}{q(u)}
\int b(v)\lambda(\mathrm dv) \notag\\
&=
Q_b\frac{p(u)}{q(u)}.
\notag
\end{align}
Therefore
\begin{align}
\beta'(1,u)
&=
I_0+I_1 \notag\\
&\leq
Q_b\frac{p(u)}{q(u)}
+
p(u)
\int
\left(
\frac{q(v)}{q(u)}
-
\frac{p(v)}{p(u)}
\right)_+
\lambda(\mathrm dv) \notag\\
&=
\beta(u),
\notag
\end{align}
which completes the proof.

\section{Additional Experiments}\label{app:futher_exps}
\paragraph{Riemann manifold Metropolis-adjusted Langevin}
Beyond random-walk Metropolis--Hastings, we also study a geometry-dependent
proposal that adapts to local structure of the target. In particular, we
consider a Riemannian manifold MALA-type proposal \cite{girolami2011riemann} on a two-dimensional banana
distribution, which is defined by the density
\[
\pi(x_1,x_2)
\propto
\exp\left\{
-\frac{x_1^2}{2\sigma_1^2}
-\frac{\bigl(x_2 + b(x_1^2-\sigma_1^2)\bigr)^2}{2\sigma_2^2}
\right\},
\qquad x=(x_1,x_2)\in\mathbb R^2,
\]
where \(b>0\) controls the curvature of the banana shape. In our experiment, we set $b=0.05$.

\begin{wrapfigure}{r}{0.45\textwidth}
\vspace{-10pt}
\centering
\begin{tikzpicture}
\begin{axis}[
    width=0.45\textwidth,
    height=0.35\textwidth,
    xlabel={Number of chains $C$},
    ylabel={Estimated $\mathrm{E}[\tau]$},
    label style={font=\scriptsize},
    tick label style={font=\scriptsize},
    xtick={2,4,8,16,32},
    ymajorgrids=true,
    grid style=dashed,
    xmode=log,
    log basis x={2},
    thick,
    mark options={scale=1.0},
    legend style={
        font=\scriptsize,
        draw=none,
        fill=none,
        at={(0.03,0.97)},
        anchor=north west
    },
    legend cell align={left}
]

\addplot[color=blue, mark=o]
coordinates {
    (2,87.19) (4,158) (8,222) (16,304) (32,395)
};
\addlegendentry{2-step star}

\addplot[color=red, mark=square*]
coordinates {
    (2,82.2) (4,150) (8,218) (16,297) (32,382)
};
\addlegendentry{1-step star}

\addplot[color=orange, mark=triangle*]
coordinates {
    (2,90) (4,192) (8,242) (16,270) (32,296)
};
\addlegendentry{2-step PMC}

\addplot[color=green!60!black, mark=diamond*]
coordinates {
    (2,86) (4,173) (8,226) (16,255) (32,278)
};
\addlegendentry{1-step PMC}

\end{axis}
\end{tikzpicture}
\vspace{-8pt}
\caption{Meeting time versus the number of coupled chains for Riemann MMALA.}
\label{fig:meeting-time-wrap-riemann}
\vspace{-14pt}
\end{wrapfigure}
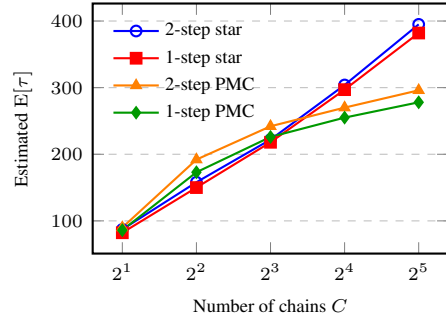 The Riemannian
proposal at the current state \(x\) is
\[
y\sim
\mathcal N\left(
x+\frac{\sigma^2}{2}G(x)^{-1}\nabla \log \pi(x),
\,
\sigma^2 G(x)^{-1}
\right),
\]
 where \(G(x)=-\nabla^2\log \pi(x)\) is the local metric tensor.  This proposal adapts to the local curvature of the target and is followed by
the standard Metropolis--Hastings accept--reject step. In our experiment, we set
\(\sigma=0.4\) and initialize the chains from
\(\pi_0=\mathrm{Unif}([-2,2]^d)\). As shown in
Figure~\ref{fig:meeting-time-wrap-riemann}, we observe the same trend as in
Section~\ref{sec:experiments}: Poisson matching outperforms the pairwise
baselines in terms of average meeting time. Furthermore, the gap between the two-step and one-step methods is smaller than
in the random-walk experiments. We attribute this to the higher acceptance rate
in this regime, which reduces the difference between coupling proposals before
and after the Metropolis--Hastings accept--reject step. Overall, this experiment
suggests that our method can be adapted beyond random-walk proposals to
geometry-dependent settings.
\subsection{Convergence Diagnostic}
We now provide an application of our method for problems in convergence diagnostic. 

\subsubsection{Grand-Coupling Diagnostic}
This line of work follows the discussion in Section \ref{sec:mh_pml} in the main text. First, we restate the classical results by \citet{johnson1996studying}. 

\begin{lemma}[\citet{johnson1996studying}]\label{lem:johnson_bound}
Let \(\kappa\) be a Markov kernel with invariant distribution \(\pi\), and let
\(X_0^{(1)},\ldots,X_0^{(C)} \overset{\mathrm{iid}}{\sim} \pi_0\). Define
\[
\omega \coloneqq \sup\{a\in[0,1]: \pi_0 \ge a\pi\}.
\]
For a grand coupling of the \(C\) chains, let \(\tau\) be the first time at
which all chains coalesce. Then
\[
\|\pi_t-\pi\|_{\mathrm{TV}}
\le
\frac{\Pr(\tau>t)}{1-(1-\omega)^C},
\]
regardless of the coupling strategy.
\end{lemma}
We observe that for sufficiently large \(C\), the denominator
\(1-(1-\omega)^C\) is close to \(1\). In this regime, the meeting-time tail
\(\Pr(\tau>t)\) becomes a direct proxy for the distributional bias
\(\|\pi_t-\pi\|_{\mathrm{TV}}\). This also highlights why faster meeting times
are desirable: the bound is controlled by the tail probability \(\Pr(\tau>t)\).

However, this diagnostic depends on knowledge of \(\omega\). In many settings,
\(\omega\) is not easy to determined and may be zero, making the bound vacuous. For example, if \(\pi_0\) is
Gaussian while \(\pi\) is heavy-tailed, then \(\pi_0\) cannot dominate any
positive multiple of \(\pi\) uniformly over the state space. We now provide an alternative bound using the idea of list-level coupling.

\begin{lemma}\label{lem:new_grand_bound}
In the setting stated above, for any faithful grand coupling of the \(C\) chains,
\[
\|\pi_t-\pi\|_{\mathrm{TV}}
\le
1-\alpha_C+\Pr(\tau>t),
\]
where \(\alpha_C\) is an achievable probability of the initial list-level
inclusion event
\[
\{Y_0\in\{X_0^{(1)},\ldots,X_0^{(C)}\}\},
\qquad Y_0\sim\pi,\quad X_0^{(i)}\sim\pi_0.
\]
Using list-level Poisson matching \citep{li2021unified,rowan2026one}, one may
take
\[
\alpha_C
\ge
\mathbb{E}_{X\sim \pi_0}\left[
\frac{C\pi(X)}{C\pi_0(X)+\pi(X)}
\right].
\]
\end{lemma}
The proof for this Lemma is at the end of this part. We note that the lower bound on \(\alpha_C\) can be estimated by Monte Carlo,
since it is an expectation under \(X\sim\pi_0\). Moreover, this quantity is
well-defined and bounded between \(0\) and \(1\). While this form is convenient
in a wider range of scenarios, it introduces an additive term in the upper
bound; improving this term, for instance using ideas from
\citet{biswas2019estimating}, is an interesting direction for future work. We now provide some examples of convergence diagnostic analysis using the RWMH examples in Section \ref{sec:experiments}.

\paragraph{Student-\(t\) RWMH.}
This setting considers a heavy-tailed Cauchy target sampled using a
Student-\(t\) random-walk proposal with \(2\) degrees of freedom, where the
chains are initialized from \(\pi_0=\mathcal N(0,\mathrm I_d)\). Note that in
this case, Lemma~\ref{lem:johnson_bound} is not applicable because the relevant
rejection coefficient is unbounded, as discussed above. We provide the estimated
bias in Figure~\ref{fig:johnson_vs_ours_gaussian} (left), where we test dimension \(d=1\).
Overall, the results suggest that using a sufficiently large number of chains
can substantially improve the diagnostic bound in this heavy-tailed setting;
among the tested values, \(C=16\) gives the tightest bound by reducing the
additive bias term associated with \(\alpha_C\) in
Lemma~\ref{fig:johnson_vs_ours_gaussian} (left). We note, however, that \(\alpha_C\) need not
approach \(1\) even for large \(C\); for example, if \(\pi_0\) is uniform on a
bounded set while \(\pi\) has mass outside this support, then the list-level
inclusion probability is bounded away from \(1\). Thus this bias may persist,
but Lemma~\ref{lem:new_grand_bound} still provides a useful heuristic for
understanding the behavior of the coupled chains, and additional chain-level
statistics may be used to further reduce or calibrate it. For example, one can
switch to another diagnostic method when the bound becomes saturated. We
illustrate this idea in the next experiment.

\begin{table}[t]
\centering
\caption{Comparison between the additive bias term \(1-\alpha_C\) from our
list-level bound and the Johnson denominator \(1-(1-\omega)^C\) for
\(\pi=\mathcal N(0,I_d)\) and \(\pi_0=\mathcal N(\mathbf 1_d,16I_d)\).}
\label{tab:alpha_vs_johnson}

\begin{minipage}[t]{0.32\linewidth}
\centering
\textbf{\(d=1\)}
\vspace{0.3em}

\begin{tabular}{c c c}
\hline
\(C\) & \(1-\alpha_C\) & Johnson \\
\hline
2   & 0.5673 & 0.4251 \\
8   & 0.2627 & 0.8908 \\
16  & 0.1526 & 0.9881 \\
32  & 0.0846 & 0.9999 \\
64  & 0.0439 & 1.0000 \\
128 & 0.0242 & 1.0000 \\
\hline
\end{tabular}
\end{minipage}
\hfill
\begin{minipage}[t]{0.32\linewidth}
\centering
\textbf{\(d=2\)}
\vspace{0.3em}

\begin{tabular}{c c c}
\hline
\(C\) & \(1-\alpha_C\) & Johnson \\
\hline
2   & 0.7468 & 0.1135 \\
8   & 0.4748 & 0.3824 \\
16  & 0.3289 & 0.6186 \\
32  & 0.2039 & 0.8546 \\
64  & 0.1149 & 0.9788 \\
128 & 0.0653 & 0.9996 \\
\hline
\end{tabular}
\end{minipage}
\hfill
\begin{minipage}[t]{0.32\linewidth}
\centering
\textbf{\(d=3\)}
\vspace{0.3em}

\begin{tabular}{c c c}
\hline
\(C\) & \(1-\alpha_C\) & Johnson \\
\hline
2   & 0.8538 & 0.0281 \\
8   & 0.6661 & 0.1077 \\
16  & 0.5292 & 0.2037 \\
32  & 0.3934 & 0.3660 \\
64  & 0.2619 & 0.5980 \\
128 & 0.1549 & 0.8384 \\
\hline
\end{tabular}
\end{minipage}

\end{table}
\paragraph{Gaussian RWMH.}
We recall the Gaussian target setting
\(\pi=\mathcal N(0,\mathrm I_d)\), with chains initialized from
\(\pi_0=\mathcal N(\mathbf 1_d,16\mathrm I_d)\). This setting is more tractable
and allows us to illustrate how one can combine the bounds from
Lemma~\ref{lem:johnson_bound} and Lemma~\ref{lem:new_grand_bound}. First,
Table~\ref{tab:alpha_vs_johnson} shows how the additive bias term
\(1-\alpha_C\) in Lemma~\ref{lem:new_grand_bound} and Johnson's denominator
\(1-(1-\omega)^C\) vary with the number of chains \(C\). As expected in this
example, \(1-\alpha_C\) decreases toward \(0\) as \(C\) increases, while
Johnson's denominator increases toward \(1\).

Figure~\ref{fig:johnson_vs_ours_gaussian} (right) reports the resulting bounds
for \(d=3\) and \(C\in\{32,64\}\). In this case, \(C=64\) yields a tighter bound
for both approaches. Furthermore, the bound from
Lemma~\ref{lem:new_grand_bound} is tighter than the one from
Lemma~\ref{lem:johnson_bound} for \(t < 50\), where the additive bias term is
less restrictive than Johnson's multiplicative denominator. For larger \(t\),
Johnson's bound can become tighter as the meeting-time tail decreases. Thus, in
this scenario, it is natural to take the minimum of the two bounds rather than
relying on a single diagnostic metric.

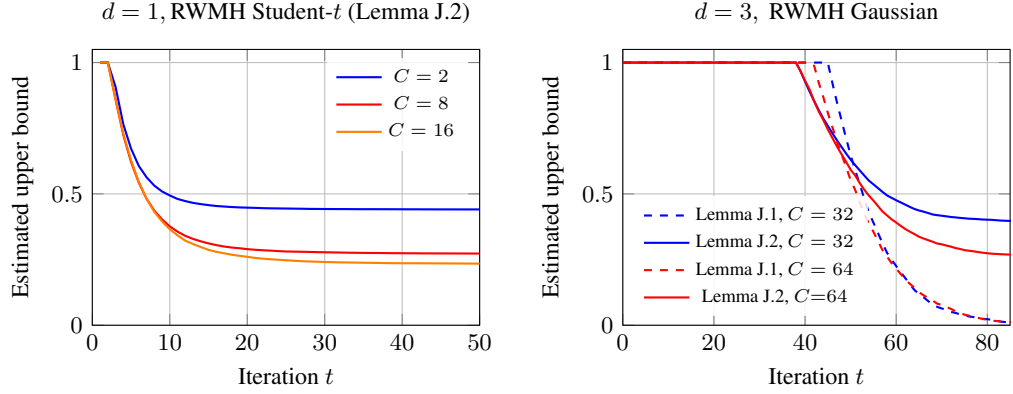
\begin{figure}[t]
\centering

\pgfplotsset{
    boundplotstyle/.style={
        width=\linewidth,
        height=0.78\linewidth,
        xlabel={Iteration $t$},
        ylabel={Estimated upper bound},
        ymin=0, ymax=1.05,
        grid=both,
        tick label style={font=\footnotesize},
        label style={font=\footnotesize},
        legend style={
            at={(0.97,0.97)},
            anchor=north east,
            font=\scriptsize,
            draw=none,
            fill=white,
            fill opacity=0.85,
            text opacity=1
        },
        title style={font=\footnotesize},
    }
}

\begin{minipage}[t]{0.48\textwidth}
\centering
\begin{tikzpicture}
\begin{axis}[
    boundplotstyle,
    xmin=0, xmax=50,
    title={\(d=1, \text{RWMH Student-$t$ (Lemma \ref{lem:new_grand_bound})}\)},
]

\addplot[thick, blue] coordinates {
(1,1.) (2,1.) (3,0.9073) (4,0.765) (5,0.6724)
(6,0.6063) (7,0.56335) (8,0.5315) (9,0.5093) (10,0.49405)
(11,0.4808) (12,0.47265) (13,0.46575) (14,0.4612) (15,0.45725)
(16,0.45395) (17,0.4519) (18,0.45) (19,0.4485) (20,0.44765)
(21,0.44685) (22,0.44575) (23,0.4451) (24,0.4447) (25,0.44435)
(26,0.44385) (27,0.44355) (28,0.4429) (29,0.44275) (30,0.44245)
(31,0.44235) (32,0.4422) (33,0.44215) (34,0.4419) (35,0.4418)
(36,0.4418) (37,0.4417) (38,0.44165) (39,0.44155) (40,0.44145)
(41,0.44135) (42,0.4413) (43,0.44115) (44,0.4411) (45,0.441)
(46,0.441) (47,0.44095) (48,0.4409) (49,0.4409) (50,0.4409)
};
\addlegendentry{$C=2$}

\addplot[thick, red] coordinates {
(1,1.) (2,1.0) (3,0.85975) (4,0.72445) (5,0.6216)
(6,0.543) (7,0.48165) (8,0.43625) (9,0.403) (10,0.3762)
(11,0.3565) (12,0.34025) (13,0.3287) (14,0.31945) (15,0.31165)
(16,0.305) (17,0.2997) (18,0.29585) (19,0.29315) (20,0.28965)
(21,0.2874) (22,0.28545) (23,0.28405) (24,0.2823) (25,0.28075)
(26,0.28025) (27,0.2797) (28,0.2791) (29,0.2783) (30,0.27775)
(31,0.2773) (32,0.2767) (33,0.27615) (34,0.2758) (35,0.27545)
(36,0.27515) (37,0.275) (38,0.2746) (39,0.2743) (40,0.27415)
(41,0.2739) (42,0.2738) (43,0.27355) (44,0.27345) (45,0.27335)
(46,0.2733) (47,0.2732) (48,0.2731) (49,0.273) (50,0.27285)
};
\addlegendentry{$C=8$}

\addplot[thick, orange] coordinates {
(1,1.) (2,1.0) (3,0.8604) (4,0.7304) (5,0.6277)
(6,0.546) (7,0.4817) (8,0.43175) (9,0.3952) (10,0.36525)
(11,0.3424) (12,0.32195) (13,0.30925) (14,0.29775) (15,0.287)
(16,0.27945) (17,0.27395) (18,0.26845) (19,0.26425) (20,0.26045)
(21,0.25645) (22,0.2534) (23,0.25115) (24,0.24925) (25,0.24745)
(26,0.24555) (27,0.2438) (28,0.24275) (29,0.24185) (30,0.24075)
(31,0.24005) (32,0.23925) (33,0.23885) (34,0.23835) (35,0.23795)
(36,0.2376) (37,0.2371) (38,0.2366) (39,0.2362) (40,0.23605)
(41,0.23605) (42,0.23595) (43,0.2357) (44,0.23565) (45,0.2355)
(46,0.2353) (47,0.2351) (48,0.23495) (49,0.2347) (50,0.23455)
};
\addlegendentry{$C=16$}

\end{axis}
\end{tikzpicture}
\end{minipage}
\hspace{0pt}
\begin{minipage}[t]{0.48\textwidth}
\centering
\begin{tikzpicture}
\begin{axis}[
    boundplotstyle,
    xmin=0, xmax=85,
    title={\(d=3, \text{ RWMH Gaussian}\)},
    legend style={
    at={(0.03,0.03)},
    anchor=south west,
    font=\scriptsize
},,
]

\addplot[thick, blue, dashed] coordinates {
(0,1) (1,1) (2,1) (3,1) (4,1) (5,1) (6,1) (7,1) (8,1) (9,1)
(10,1) (11,1) (12,1) (13,1) (14,1) (15,1) (16,1) (17,1)
(18,1) (19,1) (20,1) (21,1) (22,1) (23,1) (24,1) (25,1)
(26,1) (27,1) (28,1) (29,1) (30,1) (31,1) (32,1) (33,1)
(34,1) (35,1) (36,1) (37,1) (38,1) (39,1) (40,1) (41,1)
(42,1) (43,1) (44,1) (45,0.99900917) (46,0.92195212)
(47,0.84872060) (48,0.77986110) (49,0.71209462) (50,0.64760715)
(51,0.59022424) (52,0.53612035) (53,0.48201646) (54,0.42955208)
(55,0.39075031) (56,0.35140202) (57,0.31205374) (58,0.27543798)
(59,0.24975229) (60,0.22515962) (61,0.20220645) (62,0.17925329)
(63,0.15575362) (64,0.13443997) (65,0.11859135) (66,0.10656826)
(67,0.09235916) (68,0.07705705) (69,0.06995250) (70,0.06394096)
(71,0.05792942) (72,0.05082487) (73,0.04699934) (74,0.04153430)
(75,0.03770877) (76,0.03333674) (77,0.03115072) (78,0.02841820)
(79,0.02459268) (80,0.02295316) (81,0.02131365) (82,0.01748813)
(83,0.01366260) (84,0.01202309) (85,0.00983707) (86,0.00819756)
(87,0.00765105) (88,0.00710455) (89,0.00655805) (90,0.00601154)
(91,0.00546504) (92,0.00491854) (93,0.00437203) (94,0.00327902)
(95,0.00273252) (96,0.00218602) (97,0.00163951) (98,0.00163951)
(99,0.00163951) (100,0.00109301) (101,0.00109301)
(102,0.00054650) (103,0)
};
\addlegendentry{Lemma \ref{lem:johnson_bound}, $C=32$}

\addplot[thick, blue] coordinates {
(0,1) (1,1) (2,1) (3,1) (4,1) (5,1) (6,1) (7,1) (8,1) (9,1)
(10,1) (11,1) (12,1) (13,1) (14,1) (15,1) (16,1) (17,1)
(18,1) (19,1) (20,1) (21,1) (22,1) (23,1) (24,1) (25,1)
(26,1) (27,1) (28,1) (29,1) (30,1) (31,1) (32,1) (33,1)
(34,1) (35,1) (36,1) (37,1) (38,1) (39,0.96702965)
(40,0.92782965) (41,0.88922965) (42,0.85362965)
(43,0.81902965) (44,0.78702965) (45,0.75902965)
(46,0.73082965) (47,0.70402965) (48,0.67882965)
(49,0.65402965) (50,0.63042965) (51,0.60942965)
(52,0.58962965) (53,0.56982965) (54,0.55062965)
(55,0.53642965) (56,0.52202965) (57,0.50762965)
(58,0.49422965) (59,0.48482965) (60,0.47582965)
(61,0.46742965) (62,0.45902965) (63,0.45042965)
(64,0.44262965) (65,0.43682965) (66,0.43242965)
(67,0.42722965) (68,0.42162965) (69,0.41902965)
(70,0.41682965) (71,0.41462965) (72,0.41202965)
(73,0.41062965) (74,0.40862965) (75,0.40722965)
(76,0.40562965) (77,0.40482965) (78,0.40382965)
(79,0.40242965) (80,0.40182965) (81,0.40122965)
(82,0.39982965) (83,0.39842965) (84,0.39782965)
(85,0.39702965) (86,0.39642965) (87,0.39622965)
(88,0.39602965) (89,0.39582965) (90,0.39562965)
(91,0.39542965) (92,0.39522965) (93,0.39502965)
(94,0.39462965) (95,0.39442965) (96,0.39422965)
(97,0.39402965) (98,0.39402965) (99,0.39402965)
(100,0.39382965) (101,0.39382965) (102,0.39362965)
(103,0.39342965)
};
\addlegendentry{Lemma \ref{lem:new_grand_bound}, $C=32$}

\addplot[thick, red, dashed] coordinates {
(0,1) (1,1) (2,1) (3,1) (4,1) (5,1) (6,1) (7,1) (8,1) (9,1)
(10,1) (11,1) (12,1) (13,1) (14,1) (15,1) (16,1) (17,1)
(18,1) (19,1) (20,1) (21,1) (22,1) (23,1) (24,1) (25,1)
(26,1) (27,1) (28,1) (29,1) (30,1) (31,1) (32,1) (33,1)
(34,1) (35,1) (36,1) (37,1) (38,1) (39,1) (40,1) (41,1)
(42,0.98763106) (43,0.92341665) (44,0.86856684)
(45,0.80836582) (46,0.75418491) (47,0.70000400)
(48,0.65117429) (49,0.60301348) (50,0.55217707)
(51,0.51505311) (52,0.46722675) (53,0.42441714)
(54,0.38762763) (55,0.35585488) (56,0.32073762)
(57,0.28595481) (58,0.26120551) (59,0.23612175)
(60,0.21605475) (61,0.19464994) (62,0.17424849)
(63,0.15786043) (64,0.14247573) (65,0.12575323)
(66,0.11404747) (67,0.10401397) (68,0.09364602)
(69,0.08227472) (70,0.07725797) (71,0.06956562)
(72,0.06254216) (73,0.05351201) (74,0.04782636)
(75,0.04013401) (76,0.03645506) (77,0.03244166)
(78,0.02809381) (79,0.02374596) (80,0.02040145)
(81,0.01739140) (82,0.01605360) (83,0.01404690)
(84,0.01304355) (85,0.01170575) (86,0.01036795)
(87,0.00869570) (88,0.00836125) (89,0.00702345)
(90,0.00568565) (91,0.00468230) (92,0.00434785)
(93,0.00367895) (94,0.00301005) (95,0.00301005)
(96,0.00267560) (97,0.00234115) (98,0.00200670)
(99,0.00167225) (100,0.00167225) (101,0.00133780)
(102,0.00100335) (103,0.00066890) (104,0.00066890)
(105,0.00066890) (106,0.00066890) (107,0.00066890)
(108,0.00066890) (109,0.00033445) (110,0.00033445)
(111,0.00033445) (112,0.00033445) (113,0)
};
\addlegendentry{Lemma \ref{lem:johnson_bound}, $C=64$}

\addplot[thick, red] coordinates {
(0,1) (1,1) (2,1) (3,1) (4,1) (5,1) (6,1) (7,1) (8,1) (9,1)
(10,1) (11,1) (12,1) (13,1) (14,1) (15,1) (16,1) (17,1)
(18,1) (19,1) (20,1) (21,1) (22,1) (23,1) (24,1) (25,1)
(26,1) (27,1) (28,1) (29,1) (30,1) (31,1) (32,1) (33,1)
(34,1) (35,1) (36,1) (37,1) (38,1) (39,0.96746994)
(40,0.93086994) (41,0.89726994) (42,0.85246994)
(43,0.81406994) (44,0.78126994) (45,0.74526994)
(46,0.71286994) (47,0.68046994) (48,0.65126994)
(49,0.62246994) (50,0.59206994) (51,0.56986994)
(52,0.54126994) (53,0.51566994) (54,0.49366994)
(55,0.47466994) (56,0.45366994) (57,0.43286994)
(58,0.41806994) (59,0.40306994) (60,0.39106994)
(61,0.37826994) (62,0.36606994) (63,0.35626994)
(64,0.34706994) (65,0.33706994) (66,0.33006994)
(67,0.32406994) (68,0.31786994) (69,0.31106994)
(70,0.30806994) (71,0.30346994) (72,0.29926994)
(73,0.29386994) (74,0.29046994) (75,0.28586994)
(76,0.28366994) (77,0.28126994) (78,0.27866994)
(79,0.27606994) (80,0.27406994) (81,0.27226994)
(82,0.27146994) (83,0.27026994) (84,0.26966994)
(85,0.26886994) (86,0.26806994) (87,0.26706994)
(88,0.26686994) (89,0.26606994) (90,0.26526994)
(91,0.26466994) (92,0.26446994) (93,0.26406994)
(94,0.26366994) (95,0.26366994) (96,0.26346994)
(97,0.26326994) (98,0.26306994) (99,0.26286994)
(100,0.26286994) (101,0.26266994) (102,0.26246994)
(103,0.26226994) (104,0.26226994) (105,0.26226994)
(106,0.26226994) (107,0.26226994) (108,0.26226994)
(109,0.26206994) (110,0.26206994) (111,0.26206994)
(112,0.26206994) (113,0.26186994)
};
\addlegendentry{Lemma \ref{lem:new_grand_bound}, $C{=}64$}

\end{axis}
\end{tikzpicture}
\end{minipage}

\caption{(Left) Estimated upper bounds for $|\pi_t - \pi|_{\mathrm{TV}}$ using Lemma \ref{lem:new_grand_bound} (Right) Estimated upper bounds $|\pi_t - \pi|_{\mathrm{TV}}$ comparing Johnson's bound
\citep{johnson1996studying} and the list-level bound, i.e. Lemma \ref{lem:new_grand_bound}, in the Gaussian target
setting with \(d=3\).}
\label{fig:johnson_vs_ours_gaussian}
\end{figure}

\paragraph{Proof of Lemma \ref{lem:new_grand_bound}} We introduce the proof of Lemma \ref{lem:new_grand_bound} as follows.
\begin{proof}
Our goal is to generate \(X_0^{(1)},\ldots,X_0^{(C)}\sim \pi_0\)
independently, together with an auxiliary stationary chain \(Y_0\sim\pi\), such
that
\[
\Pr\!\left(Y_0\in\{X_0^{(1)},\ldots,X_0^{(C)}\}\right)\ge \alpha_C .
\]
To do this, we draw a Poisson point process
\[
\Pi=\{(X_i,S_i)\}_{i\ge 1},
\]
where the marks \(X_i\) have base distribution \(\pi_0\) and \(S_i\) denotes
the arrival time. We take the first \(C\) marked points as the initial chains,
\(X_0^{(1)},\ldots,X_0^{(C)}\), and select \(Y_0\) using the Poisson Monte Carlo
procedure described in Section~\ref{sec:accel_pml}. By the list-level Poisson
matching lemma \citep{li2021unified}, this construction attains an inclusion
probability \(\alpha_C\) satisfying the lower bound stated above. Other
Poisson-based coupling constructions, such as \citet{rowan2026one}, may also be
used.

We define the auxiliary chain \(Y_t\) in the following way. At time \(t\), first update
the observed list \(X_t^{1:C}\) to \(X_{t+1}^{1:C}\) using the original grand
coupling. Then:
\begin{itemize}
    \item if \(Y_t=X_t^{(i)}\) for some \(i\), choose one such index \(i^\star\)
    and set
    $Y_{t+1}=X_{t+1}^{(i^\star)};$
    \item if \(Y_t\notin\{X_t^{(1)},\ldots,X_t^{(C)}\}\), draw
    $Y_{t+1}\sim \kappa(\cdot | Y_t)$
    independently of the observed update.
\end{itemize}
Then, under this construction, we have $Y_{t+1} \sim \kappa(\cdot|Y_t)$ and thus $Y_t$ is at stationary. We now prove the event implication. On
\[
A=\{Y_0\in\{X_0^{(1)},\ldots,X_0^{(C)}\}\},
\]
choose an index \(I\) such that
\[
Y_0=X_0^{(I)}.
\]
We claim that
\[
Y_s=X_s^{(I)}
\qquad
\text{for all }s\ge 0.
\]
This holds at \(s=0\). If \(Y_s=X_s^{(I)}\), then \(Y_s\) belongs to the current
list. The auxiliary update chooses some index \(i^\star\) with
\(
Y_s=X_s^{(i^\star)}
\), hence
\(X_s^{(i^\star)}=X_s^{(I)}.
\)
By faithfulness of the observed grand coupling,
\[
X_{s+1}^{(i^\star)}=X_{s+1}^{(I)}.
\]
Since \(Y_{s+1}=X_{s+1}^{(i^\star)}\), we get
$Y_{s+1}=X_{s+1}^{(I)}.$

Therefore \(Y_s=X_s^{(I)}\) for all \(s\). If also \(\tau\le t\), then all observed chains have coalesced by time \(t\).
Thus for any fixed \(j\),
\[
X_t^{(j)}=X_t^{(I)}=Y_t.
\]
Therefore,
\[
A\cap\{\tau\le t\}
\subseteq
\{X_t^{(j)}=Y_t\},
\]
or equivalently,
\[
\{X_t^{(j)}\neq Y_t\}
\subseteq
A^c\cup\{\tau>t\}.
\]

Since \(X_t^{(j)}\) has law
\(\pi_t\) and \(Y_t\sim\pi\), the coupling inequality gives
\[
\|\pi_t-\pi\|_{\mathrm {TV}}
\le
\Pr(X_t^{(j)}\neq Y_t).
\]
Combining this with the event inclusion gives
\[
\|\pi_t-\pi\|_{\mathrm {TV}}
\le
\Pr(A^c\cup\{\tau>t\}).
\]
Finally,
\[
\Pr(A^c\cup\{\tau>t\})
\le
\Pr(A^c)+\Pr(\tau>t)
=
1-\alpha_C+\Pr(\tau>t).
\]
\end{proof}
\subsubsection{Weight-Harmonization}\label{weight_harmonization}
We show that coupling multiple chains can potentially be used to extend and improve the weight-harmonization framework for convergence diagnostics \cite{corenflos2025coupling}. Weight harmonization starts from \(N\) parallel MCMC chains, each associated with an importance weight \(W_0^{(i)}\) that reflects the discrepancy between the initial distribution \(\hat{\pi}_0\) and the target distribution \(\pi\). The basic intuition is that, as the empirical distribution \(\hat{\pi}_t\) approaches \(\pi\), the empirical distribution of the weights \(W_t^{(i)}\) should become closer to uniform. Since \(\hat{\pi}_t\) is itself unknown, the main challenge is to design an interactive update rule that propagates these weights consistently over time.

In \cite{corenflos2025coupling}, this is achieved through pairwise couplings together with partner exchanges across chains. Unlike grand coupling, where the goal is for all chains to coalesce into a single state, weight harmonization uses meetings between chains only as a mechanism for updating and propagating weights, after which the chains are randomly reassigned to new partners. Here, we show that, instead of coupling chains only in pairs, one can couple multiple chains jointly to improve diagnostic performance. The motivation is that, under pairwise couplings, two selected chains may be far apart, making coalescence unlikely and thus slowing down the weight updates. By allowing several chains to interact simultaneously, the joint coupling construction creates more opportunities for partial coalescence among nearby chains. This can lead to faster weight propagation and a more effective convergence diagnostic. We show the procedure in Algorithm \ref{alg:weight-harmonization}  where replacing group size $m=2$ recovers the proposed algorithm by \citet{corenflos2025coupling}. We note that the proposed procedure naturally satisfies the criteria for convergence diagnostics, i.e. Proposition 2 and Theorem 2 in \cite{corenflos2025coupling}.

\paragraph{Experiment.} We follow the setup in \cite{corenflos2025coupling}, where the target distribution is $\pi = \mathcal{N}(0, \mathrm{I}_d)$. We set $d=20$ and use $C=10000$ chains for convergence diagnostics. The initial distribution is $\pi_0 = \mathcal{N}(10\,\mathbf{1}_d, 5\mathrm{I}_d)$, and we consider the autoregressive kernel
\[
\kappa(\mathrm{d}y \mid x)
= \mathcal{N}(y;\rho x, (1-\rho^2)\mathrm{I}_d)\,\mathrm{d}y,
\]
since in this case the marginal law at iteration $t$ admits a closed-form expression,
\[
\hat{\pi}_t
= \mathcal{N}\!\left(\rho^t 10\,\mathbf{1}_d,\,
\bigl(1+4\rho^{2t}\bigr)\mathrm{I}_d\right).
\]
This allows us to evaluate the coupling schemes in a setting where the transient distribution is available exactly, and hence to directly assess the tightness of the bound. As the discrepancy measure, we use the squared Hellinger distance.

The performance of the different estimation schemes is shown in Figure \ref{fig:wh_graph}. We observe that our approach consistently outperforms the baseline. We attribute this improvement to the order-invariant coupling property of the Poisson Matching  framework, which avoids the asymmetry of reference-based constructions and enables more effective matching across chains.


\begin{figure}[t]
\centering
\begin{tikzpicture}
\begin{axis}[
    width=0.9\linewidth,
    height=0.35\linewidth,
    xlabel={time (iteration)},
    ylabel={sq-Hellinger},
    grid=both,
    legend pos=north east,
    legend style={font=\footnotesize},
    thick,
    mark options={scale=0.8},
    ymin=0,
    ymax=1.05,
    xmax=200,
    xmin=-1
]

\addplot[
    color=blue,
]
coordinates {
(1,0.99) (2,0.99) (3,0.99) (4,0.99) (5,0.99) (6,0.99) (7,0.99) (8,0.99)
(9,0.99) (10,0.99) (11,0.99) (12,0.99) (13,0.99) (14,0.97551) (15,0.97551)
(16,0.97551) (17,0.97551) (18,0.97046) (19,0.97046) (20,0.97046) (21,0.97046)
(22,0.97046) (23,0.97046) (24,0.96747) (25,0.94995) (26,0.94247) (27,0.91841)
(28,0.9134) (29,0.88688) (30,0.88458) (31,0.87904) (32,0.8748) (33,0.86098)
(34,0.85054) (35,0.83662) (36,0.82529) (37,0.81434) (38,0.79342) (39,0.76556)
(40,0.74473) (41,0.7294) (42,0.71144) (43,0.67965) (44,0.65771) (45,0.6314)
(46,0.60336) (47,0.55907) (48,0.53411) (49,0.50723) (50,0.48358) (51,0.46905)
(52,0.44983) (53,0.41503) (54,0.38556) (55,0.36253) (56,0.34549) (57,0.33583)
(58,0.31443) (59,0.29579) (60,0.27389) (61,0.25609) (62,0.231) (63,0.21747)
(64,0.19613) (65,0.18612) (66,0.17521) (67,0.16332) (68,0.15266) (69,0.14077)
(70,0.13229) (71,0.11615) (72,0.10728) (73,0.10035) (74,0.094629) (75,0.087023)
(76,0.082419) (77,0.076702) (78,0.071591) (79,0.065693) (80,0.058395) (81,0.054483)
(82,0.049735) (83,0.046996) (84,0.0427) (85,0.038642) (86,0.036845) (87,0.033653)
(88,0.030007) (89,0.028037) (90,0.025727) (91,0.023035) (92,0.020805) (93,0.018554)
(94,0.017463) (95,0.016773) (96,0.015641) (97,0.014856) (98,0.013838) (99,0.012748)
(100,0.011322) (101,0.010933) (102,0.0099859) (103,0.0092588) (104,0.0084132)
(105,0.0077605) (106,0.0072311) (107,0.0068549) (108,0.0064847) (109,0.0059538)
(110,0.0054755) (111,0.0051029) (112,0.0046992) (113,0.0043335) (114,0.0041451)
(115,0.0028976) (116,0.0026784) (117,0.0024644) (118,0.0022622) (119,0.0020859)
(120,0.0019847) (121,0.0018392) (122,0.0015808) (123,0.0015101) (124,0.0014242)
(125,0.0013157) (126,0.0012188) (127,0.0011222) (128,0.0010165) (129,0.00091009)
(130,0.00085056) (131,0.00074647) (132,0.00065118) (133,0.00061384) (134,0.00057999)
(135,0.00050244) (136,0.00046321) (137,0.00043496) (138,0.00038481) (139,0.00037015)
(140,0.00034838) (141,0.00032539) (142,0.00030773) (143,0.00025948) (144,0.00023781)
(145,0.00023035) (146,0.00019619) (147,0.00018032) (148,0.00014138) (149,0.00013149)
(150,0.00011662) (151,9.9328e-05) (152,8.2188e-05) (153,7.7704e-05) (154,7.5178e-05)
(155,7.1118e-05) (156,6.7249e-05) (157,5.9782e-05) (158,5.6458e-05) (159,5.4385e-05)
(160,5.0802e-05) (161,4.9332e-05) (162,4.7816e-05) (163,4.3833e-05) (164,4.1939e-05)
(165,4.0039e-05) (166,3.6824e-05) (167,3.4408e-05) (168,3.3198e-05) (169,3.2214e-05)
(170,3.1354e-05) (171,2.9937e-05) (172,2.8937e-05) (173,2.7899e-05) (174,2.74e-05)
(175,2.2777e-05) (176,2.2142e-05) (177,1.6917e-05) (178,1.6158e-05) (179,1.5635e-05)
(180,1.5194e-05) (181,1.502e-05) (182,1.3422e-05) (183,1.0337e-05) (184,7.3026e-06)
(185,6.753e-06) (186,6.6251e-06) (187,6.4975e-06) (188,6.2983e-06) (189,6.2007e-06)
(190,5.5543e-06) (191,5.4642e-06) (192,5.3117e-06) (193,2.559e-06) (194,2.4973e-06)
(195,2.282e-06) (196,2.1468e-06) (197,1.8729e-06) (198,1.6354e-06) (199,1.4198e-06)
(200,1.3579e-06) (201,1.2494e-06) (202,1.1038e-06) (203,1.0805e-06) (204,7.3264e-07)
(205,7.0108e-07) (206,6.4622e-07) (207,6.1141e-07) (208,5.7767e-07) (209,5.2856e-07)
(210,5.087e-07) (211,2.8564e-07) (212,2.3461e-07) (213,2.1683e-07) (214,1.8555e-07)
(215,1.6008e-07) (216,1.4935e-07) (217,1.3219e-07) (218,1.2831e-07) (219,1.1315e-07)
(220,9.8764e-08) (221,9.0365e-08) (222,8.5889e-08) (223,8.1048e-08) (224,7.5558e-08)
(225,6.645e-08) (226,6.1371e-08) (227,5.7559e-08) (228,5.3916e-08) (229,4.8122e-08)
(230,4.5968e-08) (231,4.4357e-08) (232,4.1773e-08) (233,3.5754e-08) (234,3.2327e-08)
(235,3.0293e-08) (236,2.8596e-08) (237,2.727e-08) (238,2.598e-08) (239,2.4373e-08)
(240,2.3166e-08) (241,1.6978e-08) (242,1.6342e-08) (243,1.5795e-08) (244,1.3693e-08)
(245,1.2014e-08) (246,1.1505e-08) (247,1.0552e-08) (248,1.0135e-08) (249,9.8196e-09)
(250,9.5336e-09) (251,9.3243e-09) (252,8.9877e-09) (253,8.6679e-09) (254,8.3037e-09)
(255,3.7392e-09) (256,3.4472e-09) (257,3.2084e-09)
};
\addlegendentry{Poisson Matching ($m=40$)}

\addplot[
    color=red,
]
coordinates {
(1,0.99) (2,0.99) (3,0.99) (4,0.99) (5,0.99) (6,0.99) (7,0.99) (8,0.99)
(9,0.99) (10,0.99) (11,0.99) (12,0.98586) (13,0.98586) (14,0.98586) (15,0.98586)
(16,0.98586) (17,0.98586) (18,0.98293) (19,0.98293) (20,0.98293) (21,0.98) (22,0.98)
(23,0.98) (24,0.98) (25,0.98) (26,0.98) (27,0.98) (28,0.97793) (29,0.97793) (30,0.97646)
(31,0.97439) (32,0.97439) (33,0.97439) (34,0.97439) (35,0.97293) (36,0.97146) (37,0.97146)
(38,0.96629) (39,0.96629) (40,0.96348) (41,0.96275) (42,0.95982) (43,0.95836) (44,0.95763)
(45,0.95366) (46,0.95263) (47,0.95159) (48,0.94952) (49,0.94723) (50,0.94495) (51,0.94245)
(52,0.94047) (53,0.93906) (54,0.93678) (55,0.93126) (56,0.92654) (57,0.9232) (58,0.92231)
(59,0.91904) (60,0.91378) (61,0.91201) (62,0.90785) (63,0.90369) (64,0.90015) (65,0.89496)
(66,0.89071) (67,0.88627) (68,0.8809) (69,0.87692) (70,0.86956) (71,0.86494) (72,0.85906)
(73,0.8531) (74,0.84756) (75,0.83886) (76,0.83292) (77,0.82454) (78,0.81668) (79,0.80977)
(80,0.79899) (81,0.78809) (82,0.78133) (83,0.77274) (84,0.76438) (85,0.75426) (86,0.74296)
(87,0.73254) (88,0.72015) (89,0.70867) (90,0.69544) (91,0.68463) (92,0.67122) (93,0.65758)
(94,0.64258) (95,0.63046) (96,0.61558) (97,0.60068) (98,0.58559) (99,0.56964) (100,0.55514)
(101,0.53815) (102,0.52337) (103,0.50502) (104,0.48877) (105,0.47208) (106,0.45754) (107,0.44134)
(108,0.42431) (109,0.40589) (110,0.38819) (111,0.37084) (112,0.35508) (113,0.33719) (114,0.31969)
(115,0.30458) (116,0.28891) (117,0.27386) (118,0.2589) (119,0.24371) (120,0.2297) (121,0.21584)
(122,0.20331) (123,0.19089) (124,0.17825) (125,0.16732) (126,0.1569) (127,0.14746) (128,0.13897)
(129,0.13005) (130,0.12182) (131,0.11288) (132,0.10572) (133,0.099885) (134,0.093364) (135,0.08648)
(136,0.080136) (137,0.074786) (138,0.070036) (139,0.064902) (140,0.060183) (141,0.055375) (142,0.051857)
(143,0.048458) (144,0.045333) (145,0.0416) (146,0.038887) (147,0.036314) (148,0.034063) (149,0.031724)
(150,0.029297) (151,0.027239) (152,0.025758) (153,0.023841) (154,0.022149) (155,0.020828) (156,0.019623)
(157,0.018368) (158,0.017215) (159,0.016035) (160,0.014965) (161,0.013872) (162,0.012972) (163,0.012205)
(164,0.011334) (165,0.010525) (166,0.0097701) (167,0.0091438) (168,0.0085228) (169,0.007887) (170,0.0073994)
(171,0.006855) (172,0.0063291) (173,0.0059044) (174,0.0055412) (175,0.0051734) (176,0.0048186) (177,0.0044854)
(178,0.0041991) (179,0.0039194) (180,0.0036628) (181,0.0033984) (182,0.0031618) (183,0.0029458) (184,0.002755)
(185,0.0025344) (186,0.0023659) (187,0.002197) (188,0.0020671) (189,0.0019394) (190,0.0018034) (191,0.001694)
(192,0.0015977) (193,0.0014762) (194,0.0013778) (195,0.0012781) (196,0.0011925) (197,0.0011007) (198,0.0010259)
(199,0.00095874) (200,0.00089661) (201,0.00082999) (202,0.00076709) (203,0.00071255) (204,0.00066647)
(205,0.00061577) (206,0.00057485) (207,0.00053692) (208,0.00049797) (209,0.00046681) (210,0.00043803)
(211,0.00040533) (212,0.00037707) (213,0.00035346) (214,0.00032239) (215,0.0003002) (216,0.00027675)
(217,0.00025842) (218,0.00024183) (219,0.00022211) (220,0.00020606) (221,0.00019401) (222,0.00017992)
(223,0.0001689) (224,0.00015489) (225,0.00014429) (226,0.00013549) (227,0.00012722) (228,0.00011859)
(229,0.00010856) (230,0.00010056) (231,9.2885e-05) (232,8.569e-05) (233,8.0191e-05) (234,7.4659e-05)
(235,6.9749e-05) (236,6.5485e-05) (237,6.1265e-05) (238,5.6651e-05) (239,5.37e-05) (240,4.9767e-05)
(241,4.6505e-05) (242,4.3474e-05) (243,4.0257e-05) (244,3.7511e-05) (245,3.5018e-05) (246,3.2663e-05)
(247,3.0281e-05) (248,2.8248e-05) (249,2.6249e-05) (250,2.4482e-05) (251,2.2713e-05) (252,2.127e-05)
(253,1.9906e-05) (254,1.8583e-05) (255,1.7421e-05) (256,1.6086e-05) (257,1.5039e-05)
};
\addlegendentry{\citet{corenflos2025coupling}}

\addplot[
    color=orange,
]
coordinates {
(1,1.0) (2,1.0) (3,1.0) (4,1.0) (5,1.0) (6,1.0) (7,1.0) (8,1.0)
(9,0.9987) (10,0.9871) (11,0.9404) (12,0.8378) (13,0.6896) (14,0.5282)
(15,0.3823) (16,0.2656) (17,0.1794) (18,0.1190) (19,0.0779) (20,0.0506)
(21,0.0327) (22,0.0210) (23,0.0135) (24,0.0087) (25,0.0056) (26,0.0036)
(27,0.0023) (28,0.0015) (29,0.0009) (30,0.0006) (31,0.0004) (32,0.0002)
(33,0.0002) (34,0.0001) (35,6.4254e-05) (36,4.1127e-05) (37,2.6345e-05)
(38,1.6868e-05) (39,1.0788e-05) (40,6.9141e-06) (41,4.4107e-06) (42,2.8014e-06)
(43,1.7881e-06) (44,1.1325e-06) (45,7.1526e-07) (46,4.7684e-07) (47,2.9802e-07)
(48,1.7881e-07) (49,1.1921e-07) (50,5.9605e-08) (51,5.9605e-08) (52,5.9605e-08)
(53,0.0) (54,0.0) (55,0.0) (56,0.0) (57,0.0) (58,0.0) (59,0.0) (60,0.0)
(61,0.0) (62,0.0) (63,0.0) (64,0.0) (65,0.0) (66,0.0) (67,0.0) (68,0.0)
(69,0.0) (70,0.0) (71,0.0) (72,0.0) (73,0.0) (74,0.0) (75,0.0) (76,0.0)
(77,0.0) (78,0.0) (79,0.0) (80,0.0) (81,0.0) (82,0.0) (83,0.0) (84,0.0)
(85,0.0) (86,0.0) (87,0.0) (88,0.0) (89,0.0) (90,0.0) (91,0.0) (92,0.0)
(93,0.0) (94,0.0) (95,0.0) (96,0.0) (97,0.0) (98,0.0) (99,0.0) (100,0.0)
(101,0.0) (102,0.0) (103,0.0) (104,0.0) (105,0.0) (106,0.0) (107,0.0) (108,0.0)
(109,0.0) (110,0.0) (111,0.0) (112,0.0) (113,0.0) (114,0.0) (115,0.0) (116,0.0)
(117,0.0) (118,0.0) (119,0.0) (120,0.0) (121,0.0) (122,0.0) (123,0.0) (124,0.0)
(125,0.0) (126,0.0) (127,0.0) (128,0.0) (129,0.0) (130,0.0) (131,0.0) (132,0.0)
(133,0.0) (134,0.0) (135,0.0) (136,0.0) (137,0.0) (138,0.0) (139,0.0) (140,0.0)
(141,0.0) (142,0.0) (143,0.0) (144,0.0) (145,0.0) (146,0.0) (147,0.0) (148,0.0)
(149,0.0) (150,0.0) (151,0.0) (152,0.0) (153,0.0) (154,0.0) (155,0.0) (156,0.0)
(157,0.0) (158,0.0) (159,0.0) (160,0.0) (161,0.0) (162,0.0) (163,0.0) (164,0.0)
(165,0.0) (166,0.0) (167,0.0) (168,0.0) (169,0.0) (170,0.0) (171,0.0) (172,0.0)
(173,0.0) (174,0.0) (175,0.0) (176,0.0) (177,0.0) (178,0.0) (179,0.0) (180,0.0)
(181,0.0) (182,0.0) (183,0.0) (184,0.0) (185,0.0) (186,0.0) (187,0.0) (188,0.0)
(189,0.0) (190,0.0) (191,0.0) (192,0.0) (193,0.0) (194,0.0) (195,0.0) (196,0.0)
(197,0.0) (198,0.0) (199,0.0) (200,0.0) (201,0.0) (202,0.0) (203,0.0) (204,0.0)
(205,0.0) (206,0.0) (207,0.0) (208,0.0) (209,0.0) (210,0.0) (211,0.0) (212,0.0)
(213,0.0) (214,0.0) (215,0.0) (216,0.0) (217,0.0) (218,0.0) (219,0.0) (220,0.0)
(221,0.0) (222,0.0) (223,0.0) (224,0.0) (225,0.0) (226,0.0) (227,0.0) (228,0.0)
(229,0.0) (230,0.0) (231,0.0) (232,0.0) (233,0.0) (234,0.0) (235,0.0) (236,0.0)
(237,0.0) (238,0.0) (239,0.0) (240,0.0) (241,0.0) (242,0.0) (243,0.0) (244,0.0)
(245,0.0) (246,0.0) (247,0.0) (248,0.0) (249,0.0) (250,0.0) (251,0.0) (252,0.0)
(253,0.0) (254,0.0) (255,0.0) (256,0.0) (257,0.0)
};
\addlegendentry{Theoretical (ground-truth)}

\end{axis}
\end{tikzpicture}
\vspace{-10pt}
\caption{Squared Hellinger distance across the chains for the different estimation schemes.}
\label{fig:wh_graph}
\end{figure}
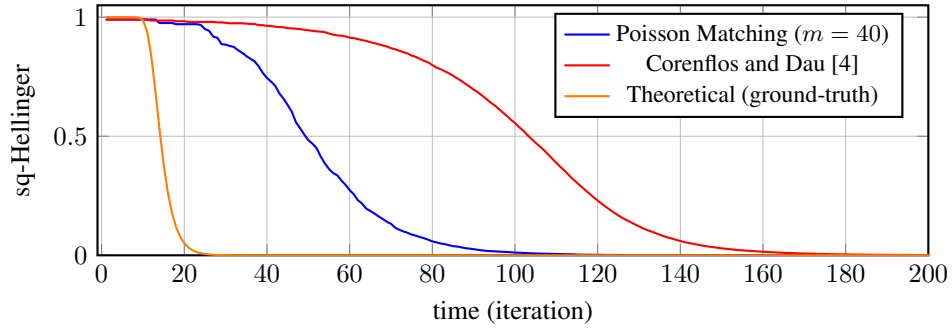
\begin{algorithm}[ht]
\SetAlgoLined
\DontPrintSemicolon
\SetKwInOut{Input}{Input}\SetKwInOut{Output}{Output}
\SetKwFunction{GrandCouple}{GrandCouple}
\SetKwFunction{Average}{Average}
\SetKwFunction{FindClusters}{FindClusters}
\SetKwFunction{Reshuffle}{Reshuffle}

\Input{Chain states $\{X_t^{(i)}\}_{i=1}^N$, weights $\{W_t^{(i)}\}_{i=1}^N$, group size $m$, and a grand-coupling transition rule}
\Output{Updated chain states $\{X_{t+1}^{(i)}\}_{i=1}^N$, updated weights $\{W_{t+1}^{(i)}\}_{i=1}^N$}
\;

Partition the chains into groups
\[
\mathcal{G}_t = \{G_t^{(1)},\dots,G_t^{(L)}\},
\qquad |G_t^{(\ell)}| = m \text{ for all } \ell
\text{ where $L = N/m$}\]
$\mathcal{A}_t \gets \emptyset$\tcp*{indices of groups with at least one coalescence}

\For{$\ell = 1,\dots,L$}{
    Let $G_t^{(\ell)} = \{i_1,\dots,i_m\}$\;
    
    Sample the joint transition:
    \(
    (X_{t+1}^{(i_1)},\dots,X_{t+1}^{(i_m)})
    \gets
    \GrandCouple(X_t^{(i_1)},\dots,X_t^{(i_m)})
    \)
    \;
    
    Find all coalesced clusters:
    \(
    \mathcal{C}_\ell \gets \FindClusters(X_{t+1}^{(i_1)},\dots,X_{t+1}^{(i_m)})
    \)
    \;
    
    \ForEach{cluster $C \in \mathcal{C}_\ell$ such that $|C| \geq 2$}{
        $\bar{W} \gets \Average\big(\{W_t^{(j)} : j \in C\}\big)$\;
        \ForEach{$j \in C$}{
            $W_{t+1}^{(j)} \gets \bar{W}$\;
        }
    }
    
    \ForEach{$j \in G_t^{(\ell)}$ not belonging to any cluster $C \in \mathcal{C}_\ell$ with $|C|\geq 2$}{
        $W_{t+1}^{(j)} \gets W_t^{(j)}$\;
    }
    
    \If{there exists a cluster $C \in \mathcal{C}_\ell$ with $|C|\geq 2$}{
        $\mathcal{A}_t \gets \mathcal{A}_t \cup \{\ell\}$\;
    }
}
\;

\If{$|\mathcal{A}_t| \geq 2$}{
    Reshuffle the groups indexed by $\mathcal{A}_t$ to form the next grouping $\mathcal{G}_{t+1}$\;
}
\Else{
    $\mathcal{G}_{t+1} \gets \mathcal{G}_t$\;
}
\;

\caption{Weight harmonization with groupwise grand coupling}
\label{alg:weight-harmonization}
\end{algorithm}


\end{document}